\documentclass[journal,comsoc]{IEEEtran}

\usepackage[T1]{fontenc}
\usepackage{subfigure}
\usepackage{tabularx}
\usepackage{xcolor}
\usepackage{multirow}
\usepackage{booktabs}
\usepackage{array}
\usepackage{amsmath,amssymb,amsfonts,bm}
\usepackage[english]{babel}
\usepackage{amsthm}
\usepackage{cite}
\usepackage[pdftex]{graphicx}
\usepackage[normalem]{ulem}
\usepackage[cmintegrals]{newtxmath}
\usepackage{algorithmic}
\usepackage{multicol}
\usepackage{algorithm}
\usepackage{verbatim}
\usepackage{makecell}
\usepackage[pagewise]{lineno}
\usepackage{upgreek}

\newcolumntype{L}[1]{>{\raggedright\let\newline\\\arraybackslash\hspace{0pt}}m{#1}}
\newcolumntype{C}[1]{>{\centering\let\newline\\\arraybackslash\hspace{0pt}}m{#1}}
\newcolumntype{R}[1]{>{\raggedleft\let\newline\\\arraybackslash\hspace{0pt}}m{#1}}

\theoremstyle{plain}
\newtheorem{theorem}{Theorem}
\newtheorem{lemma}{Lemma}

\newtheorem{proposition}{Proposition}
\theoremstyle{remark}
\newtheorem{remark}{Remark}

\newcommand{\keff}{\kappa_2^{\mathrm{eff}}}

\begin{document}
\abovedisplayshortskip=1pt
\belowdisplayshortskip=1pt
\abovedisplayskip=1pt
\belowdisplayskip=1pt
\textfloatsep=1pt
\floatsep=1pt
\intextsep=1pt

\setcounter{figure}{0}
\renewcommand{\figurename}{Fig.}
\renewcommand{\thefigure}{\arabic{figure}}

\title{Full-Frame OFDM-ISAC: Joint Constellation, Pilot, and Receive-Filter Design}

\author{Dongil~Yang,~\IEEEmembership{Student Member,~IEEE,}
Kaitao~Meng,~\IEEEmembership{Member,~IEEE,}
Fan~Liu,~\IEEEmembership{Senior Member,~IEEE,}
Christos~Masouros,~\IEEEmembership{Fellow,~IEEE,} and
Kawon~Han,~\IEEEmembership{Member,~IEEE}
\thanks{D. Yang and K. Han are with the Department of Electrical
Engineering, Ulsan National Institute of Science and Technology (UNIST),
Ulsan, South Korea (emails:\{dongil1223, kawon.han\}@unist.ac.kr).}
\thanks{K. Meng is with the Department of Electrical and Electronic
Engineering, University of Manchester, Manchester, UK.}
\thanks{F. Liu is with the National Mobile Communications Research Laboratory, School of Information Science and Engineering, Southeast University, Nanjing, China.}
\thanks{C. Masouros is with the Department of Electronic and Electrical
Engineering, University College London, London, UK.}}

\maketitle

\begin{abstract}
Orthogonal frequency-division multiplexing (OFDM) integrated sensing and communication (ISAC) reuses the entire time-frequency frame, including pilots and data payloads, for radar sensing. The random data payload, however, reduces the sensing dynamic range and degrades the target detection and parameter estimation performance. Existing studies typically optimize the modulation constellation, the pilot pattern, the power allocation, and the receive filter in isolation, overlooking their coupled effects on the sensing ambiguity structure. To address this gap, we propose a joint design of the constellation, the pilot placement and power split, and the mismatched receive filter for full-frame OFDM-ISAC systems, which minimizes the residual interference inside declared regions of interest (ROIs) under a communication rate constraint. For a given transmit signaling scheme, the receive filter admits a closed-form solution. With this filter in place, we derive closed-form expressions for the receiver-specific interference floor and noise enhancement over the 2-D range-Doppler domain, which show that the floor is governed by a single effective statistic of the full-frame signaling. This law reduces the joint design to two steps. A balanced staggered pilot placement removes the deterministic grating lobes of pilot-power boosting, and the shaped constellation and the power split are then selected jointly on the boundary of the rate constraint under the channel estimation error, which attains the joint optimum over the constructed shaping levels to first order in the declared ROI size. Under a quality-of-service (QoS) constraint of $98\%$ of the rate of a conventional 16-QAM frame, the proposed design significantly lowers the interference floor, improves the target detection probability, and achieves better ranging accuracy compared to the conventional receive filter-aware design.
\end{abstract}


\IEEEpeerreviewmaketitle

\vspace{-1\baselineskip}
\section{Introduction}

\IEEEPARstart{I}{ntegrated} sensing and communication (ISAC) has emerged as a key enabling technology for future wireless networks, in which the communication infrastructure is expected not only to provide high-rate connectivity but also to acquire environmental information for localization, tracking, mapping, and situational awareness \cite{liu2022integrated,liu2020joint0,zhang2021overview,luo2025isac}. Among candidate waveforms, orthogonal frequency-division multiplexing (OFDM) is particularly attractive for ISAC due to its robustness to frequency-selective channels, its compatibility with existing cellular standards, and its simple delay-Doppler processing through two-dimensional Fourier transforms \cite{sturm2011waveform,liu2025cp}. These properties allow an OFDM transceiver to reuse the same time-frequency resources for both communication and sensing, thereby avoiding the spectral inefficiency of dedicated radar transmissions.

A communication-centric OFDM-ISAC frame contains not only deterministic pilot symbols but also random data payloads. Conventional sensing approaches typically rely on the deterministic known pilots, whose regular structure allows simple deterministic processing \cite{zhu2023pilot,hua2024integrated}. Such pilot-only sensing, however, discards most of the signaling resources as data payloads occupy a large portion of the frame. For example, in a 5G NR PDSCH configuration for high-mobility links, with a normal cyclic prefix and a single-symbol Type-1 DM-RS with three additional DM-RS positions, pilots occupy approximately $14\%$ of the resource elements in each physical resource block \cite{3gpp38211}. In contrast, payload-based sensing can utilize more OFDM resources, which are intrinsically available at a monostatic sensing receiver. This substantially increases the coherent integration gain without imposing additional sensing overheads, and is therefore especially attractive for communication-centric ISAC \cite{du2024reshaping,liu2025uncovering,han2026constellation,xu2025exploiting, han2026next}.

The use of random payloads nevertheless introduces a fundamental limitation compared to sensing-centric waveforms. Unlike a constant-modulus radar sequence, a practical quadrature amplitude modulation (QAM) or amplitude-phase-shift keying (APSK) payload has symbol-dependent amplitude fluctuations. They appear as data-dependent sidelobes of the ambiguity function around every scatterer in the range-Doppler response \cite{liu2025cp,lu2024random,liu2025uncovering}. Although the matched filtering (MF) is optimal for signal-to-noise ratio (SNR), its output retains a payload-dependent sidelobe floor whose level is dictated by the fourth-order moment of the constellation \cite{han2026constellation,meng2026constellation}. As a result, a dominant clutter return produces a broad interference floor that masks targets regardless of the noise level.

Several sensing receive filters have been considered to address this problem. The reciprocal filter (RF) removes the payload modulation by exactly inverting the transmitted symbols \cite{wojaczek2018reciprocal,rodriguez2023supervised, han2025sensing}, suppressing the data-dependent sidelobes at the cost of an SNR loss due to amplified noise, of which performance is governed by the inverse second-order moment of the constellation \cite{meng2026constellation} and becomes substantial when constellation points approach the origin. On the other hand, the Wiener filter is the linear minimum mean-square error (LMMSE) estimator of the sensing channel response. It regularizes the symbol inversion by the target SNR and thereby interpolates between the MF and the RF, trading residual data-dependent sidelobes against noise enhancement \cite{rodriguez2023supervised,keskin2025fundamental}. Importantly, neither filter chooses where the sidelobe suppression is aimed, and the SNR loss is incurred over the whole frame although detection depends only on the range-Doppler cells that cover the region of interest (ROI). Within the same family, the region-of-interest mismatched filter (ROI-MMF) \cite{yang2026constellation} concentrates its degrees of freedom on the delay region in which weak targets are expected, preserves the mainlobe response. This approach has been shown that the constellation dependency for range estimation is reduced when the ROI is small enough.

In parallel, a body of work on receiver-aware transmit design has tuned the modulation geometry or its input distribution to the deployed sensing receiver. Specifically, these designs match the fourth-order moment that sets the MF sidelobe level \cite{han2026constellation,geiger2026constellation}, the inverse second-order moment that governs the RF noise enhancement \cite{meng2026constellation}, or the estimation accuracy attained by the receiver \cite{yang2026constellation}. The resulting sensing-rate trade-off has been characterized for random waveforms \cite{yang2024constellation} and for finite-constellation probabilistic shaping \cite{liu2025probabilistic}. These designs optimize the payload constellation alone, however, and treat the pilot structure as given. On the pilot side, boosting the pilot power improves channel estimation but reduces the power available to the payload symbols \cite{zhu2023pilot,hua2024integrated}. Moreover, a regular pilot comb facilitates channel estimation and deterministic pilot-based sensing \cite{sturm2011waveform,keskin2021mimo}, yet under pilot-power boosting its periodic power pattern imposes deterministic grating artifacts on the range-Doppler domain. The strong influence of the pilot pattern on the OFDM sensing response has motivated recent pilot-pattern designs \cite{zhang2024cross,bouziane2026optimized,wu2026ambiguity}. Beyond these separate refinements, recent studies have begun to exploit pilots and payloads jointly for sensing, demonstrating clear gains over pilot-only processing \cite{xu2025exploiting,xie2025adaptive}.

Despite these advancements, existing works generally either optimize the receive filter for a fixed transmit frame or shape the transmit signal for a fixed receiver. As a result, constellation shaping, pilot-density selection, pilot/data power allocation, pilot placement, and mismatched receive-filter design have largely been treated as separate design problems \cite{geiger2026constellation,du2024reshaping,zhu2023pilot,hua2024integrated,zhang2025optimal,liu2025uncovering,mcaulay1971optimal}. Therefore, the joint design of the transmit frame and the sensing receiver remains an open problem. The key challenge is that these design variables are tightly coupled yet act on the sensing performance through different mechanisms. The constellation and the pilot-payload configuration determine the data-dependent interference statistics and the ambiguity structure, whereas the receive filter suppresses the resulting interference at the cost of noise enhancement. Jointly optimizing them therefore gives rise to a mixed discrete-continuous and highly nonconvex problem under simultaneous communication-rate and sensing-performance requirements. Consequently, there is a compelling need for a full-frame OFDM-ISAC design that jointly accounts for the communication rate, the sensing receiver, and the complete pilot-payload frame structure.

To address the above gap, we develop a receiver-aware full-frame OFDM-ISAC framework that jointly designs the modulation constellation, the pilot structure, and the receive filter under a prescribed communication-rate requirement. Leveraging a unified characterization of the post-filter interference for the MF, the RF, and the ROI-MMF, we separate the random contribution of the data payload from the deterministic contribution of the boosted pilot lattice, connect both to the detection and estimation performance, and derive the channel-estimation-aware achievable rate that closes the design loop on the communication side. The main contributions of this paper are summarized as follows.

\begin{figure}[t!]
    \centering
    {\includegraphics[width=0.5\textwidth]{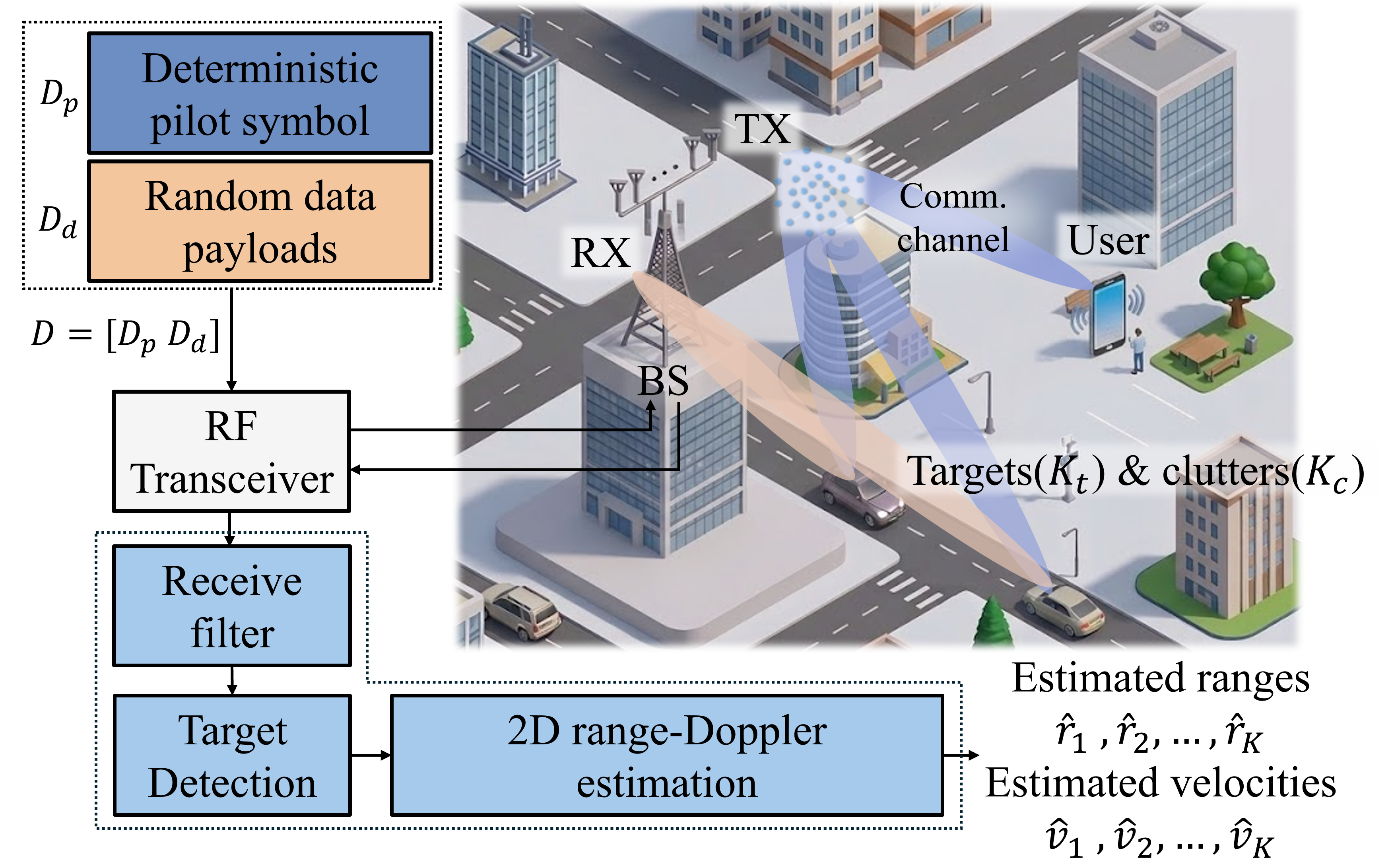}}
    \caption{System model of the communication-centric
    OFDM-ISAC system.}
    \label{f1}
\end{figure}

\begin{itemize}
\item We formulate the joint design of the modulation constellation, the pilot placement and power split, and the mismatched receive filter for full-frame OFDM-ISAC, in which the residual interference inside declared ROIs is minimized subject to a prescribed communication-rate requirement. the proposed formulation couples all four design variables, and we derive a closed-form decomposition of its objective, in which each variable acts through an identified term, namely the deterministic pilot ambiguity, the payload-fluctuation statistic, and the filter-noise energy.
\item We establish a unified receiver-aware analysis of full-frame OFDM sensing that covers both pilots and data payloads under a common mainlobe-gain normalization. Our results reveal that the payload-induced interference floor of the ROI-MMF is governed by a single statistic, the effective centered second moment of the full-frame OFDM signals, which couples the constellation, the pilot density, and the power split, and quantifies the ROI-MMF residual separately at the covered and uncovered lags.
\item Building on this analysis, we develop a three-step design procedure with an optimality guarantee. The receive filter is obtained in closed form for each transmitted frame. The pilot placement follows from a balanced staggered construction that removes the deterministic grating response caused by pilot-power boosting while preserving the orthogonality required for channel estimation. The constellation and the power split are then selected on the rate boundary through a channel-estimation-aware effective signal-to-interference-plus-noise ratio (SINR) and a quality of service (QoS)-based selection rule. We further establish that this sequential procedure attains the first-order joint optimum over the constructed shaping levels.
\end{itemize}

%The proposed framework is validated through extensive numerical simulations, in which the joint design is evaluated against the baseline under the declared QoS constraint and all shaped constellations are compared at matched mutual information. The closed-form floor predictions match the Monte Carlo measurements. The joint design lowers the payload-induced floor while preserving the QoS requirement, restores targets masked under conventional transmit and receive processing, advances the CFAR detection threshold relative to the conventional receive filters, and preserves CRB-tracking range estimation.

\textit{Notations}: Boldface lower-case and upper-case symbols denote vectors and matrices, respectively. $\mathbf{A}\in\mathbb{C}^{N\times M}$ and $\mathbf{B}\in\mathbb{R}^{N\times M}$ represent a complex-valued $N\times M$ matrix and a real-valued $N\times M$ matrix, respectively. The operators $(\cdot)^{T}$, $(\cdot)^{H}$, and $(\cdot)^{*}$ represent the transpose, Hermitian transpose, and conjugate, respectively. $\mathrm{diag}(\mathbf{a})$ denotes a diagonal matrix whose diagonal elements are given by the vector $\mathbf{a}$, and $\operatorname{vec}(\cdot)$ denotes column-wise vectorization. The operators $\odot$ and $\oslash$ represent the Hadamard product and the element-wise division, respectively. Finally $\mathbb{E}[\cdot]$ denotes the statistical expectation operator. 
\vspace{-1\baselineskip}
% =====================================================================
\section{System Model}\label{sec:sysmodel}

We consider a communication-centric OFDM-ISAC system, in which a base station (BS) simultaneously serves a communication user and senses multiple targets in a cluttered environment, as illustrated in Fig.~\ref{f1}. The communication and sensing functions share the same OFDM frame, which consists of deterministic pilot symbols and random data payloads. Since the monostatic sensing receiver is co-located with the transmitter, the transmitted frame is intrinsically known and is reused as a full-frame sensing reference. 

\begin{figure}[t!]
    \centering
    {\includegraphics[width=0.5\textwidth]{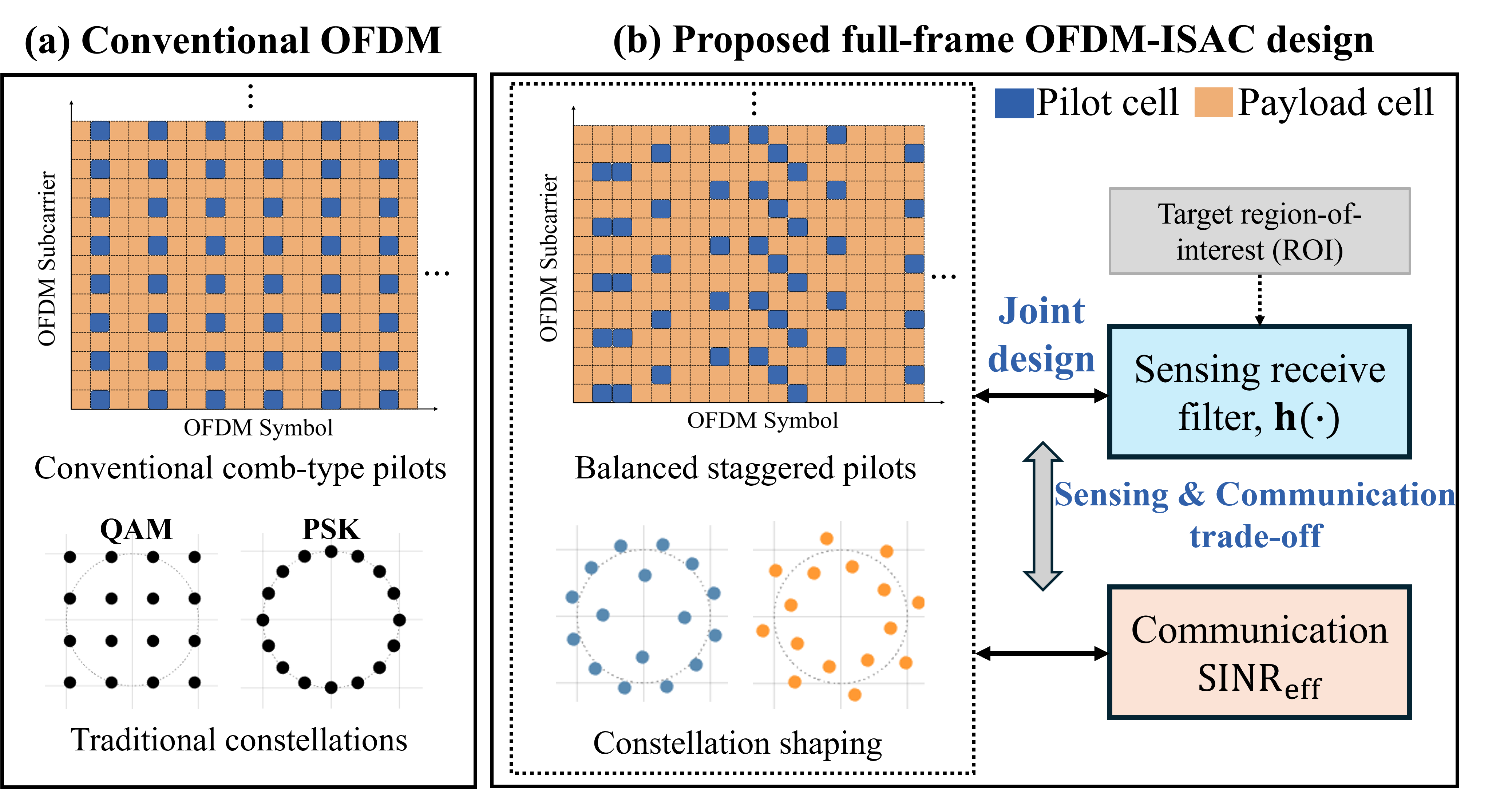}}
    \caption{(a) Conventional and (b) proposed full-frame OFDM transmit signaling structure.}
    \label{f2}
\end{figure}
\vspace{-0.7\baselineskip}
\subsection{OFDM Frame Structure}\label{sec:frame}

Consider an OFDM frame consisting of $N$ subcarriers and $M$ consecutive OFDM symbols, giving a total of $D=NM$ time-frequency resource elements. The subcarrier spacing is denoted by $\Delta f$, the bandwidth is $N\Delta f$, and the useful OFDM symbol duration is $T_{\rm u}=1/\Delta f$. With a cyclic-prefix (CP) duration $T_{\rm cp}$, the total OFDM symbol duration is $T_{\rm o}=T_{\rm u}+T_{\rm cp}$, and all delays of interest are assumed shorter than $T_{\rm cp}$, so the frame model contains only the $N\times M$ data-bearing elements free of inter-symbol interference.

Let $\mathbf B\in\{0,1\}^{N\times M}$ denote the pilot-placement matrix, where $B_{n,m}=1$ indicates that the resource element $(n,m)$ carries a pilot symbol and $B_{n,m}=0$ indicates a data symbol. The pilot and data index sets are respectively $\mathcal P=\{(n,m):B_{n,m}=1\}$ and $\mathcal D=\{(n,m):B_{n,m}=0\}$, with cardinalities $D_p=|\mathcal P|$ and $D_d=|\mathcal D|=D-D_p$. Accordingly, the pilot density is defined as
\begin{equation}
    \eta=\frac{D_p}{D}.
    \label{eq:pilot_density}
\end{equation}
Throughout, the pilots form a comb of density $\eta=1/f_p$, where $f_p$ is the comb factor. In the conventional vertical comb, one pilot occupies every $f_p$-th subcarrier of every OFDM symbol, and this structure serves as the reference placement in this paper. Beyond this reference, the placement itself is a design variable of this work, and the admissible placements are constrained to be balanced, i.e., every subcarrier carries the same number of pilots over the frame,
\begin{equation}
    \sum_{m=0}^{M-1}B_{n,m}=\frac{D_p}{N},
    \qquad n=0,\ldots,N-1.
    \label{eq:balanced_class}
\end{equation}
The resulting set of placements is denoted by $\mathcal B$. Note that the conventional vertical comb concentrates its pilots on $N/f_p$ subcarriers and therefore lies outside $\mathcal B$.

Let $\mathbf S_p\in\mathbb C^{N\times M}$ and $\mathbf S_d\in\mathbb C^{N\times M}$ denote the deterministic pilot-symbol matrix and the random data-symbol matrix, respectively. The pilot symbols carry a Zadoff-Chu (ZC) sequence, a constant-amplitude zero-autocorrelation waveform with $|[S_p]_{n,m}|=1$ for $(n,m)\in\mathcal P$. The nonzero elements of $\mathbf S_d$ are independently and equiprobably drawn from an $L$-ary constellation $\mathcal X=\{x_1,\ldots,x_L\}$, normalized such that $\mathbb E[|x|^2]=1$. 

The constellation is constrained to be invariant under a quarter-turn rotation, that is, $x\in\mathcal X$ if and only if $jx\in\mathcal X$, which is referred to as $C_4$ symmetry \cite{boutillon2024constellations}. The points therefore occur in quadruples $\{x,jx,-x,-jx\}$, and since the symbols are equiprobable, this symmetry implies
\begin{equation}
    \mathbb E[x]=0,
    \qquad
    \mathbb E[x^2]=0.
    \label{eq:c4_moments}
\end{equation}
Quarter-turn symmetric constellations are known to retain near-capacity mutual information under geometric shaping \cite{boutillon2024constellations}, so the constraint sacrifices little shaping freedom while reducing the free design variables to $L/4$ points. On the sensing side, \eqref{eq:c4_moments} reduces every payload statistic that enters the sensing analysis to a power statistic. Specifically, for two distinct cells the independence and $\mathbb E[x]=0$ give $\mathbb E[x_nx_m^{*}]=\mathbb E[x_nx_m]=0$, and within a cell $\mathbb E[x^2]=0$ removes the remaining phase term, so the only surviving second-order quantity is the power $\mathbb E[|x|^2]$. Fig.~\ref{f2} contrasts the two transmit structures. The comb-type pilots and the traditional constellations of Fig.~\ref{f2}(a) are replaced by the balanced staggered pilots and the shaped constellations of Fig.~\ref{f2}(b), and these two transmit variables are designed jointly with the ROI-aware sensing receive filter under the sensing-communication trade-off, which is the subject of this paper.

Finally, $p_p$ and $p_d$ denote the average powers assigned to a pilot and a data resource element, respectively \cite{vsimko2012optimal,simko2013adaptive}. The transmitted OFDM frame is then given by
\begin{equation}
    \mathbf X
    =\sqrt{p_p}\,\mathbf B\odot\mathbf S_p
    +\sqrt{p_d}\,(\mathbf 1-\mathbf B)\odot\mathbf S_d,
    \label{eq:transmit_frame}
\end{equation}
subject to the normalized average-power constraint
\begin{equation}
    \eta p_p+(1-\eta)p_d=1,
    \label{eq:frame_power_constraint}
\end{equation}
which guarantees $\mathbb E[\|\mathbf X\|_F^2]=D$. To parameterize the pilot/data power allocation, we define the fraction of the total frame energy assigned to the pilots as $\gamma=\eta p_p$, so that
\eqref{eq:frame_power_constraint} yields
\begin{equation}
    p_p=\frac{\gamma}{\eta},
    \qquad
    p_d=\frac{1-\gamma}{1-\eta}.
    \label{eq:pilot_data_power}
\end{equation}
The uniform-power allocation corresponds to $\gamma=\eta$, for which $p_p=p_d=1$, whereas $\gamma>\eta$ represents pilot-power boosting at the expense of the data-symbol power. Standard-compliant per-element boosting is upper-bounded by $p_p/p_d\le B_{\rm st}$, which restricts $\gamma\in[\eta,\gamma_{\max}]$ with
\begin{equation}
    \gamma_{\max}=\frac{B_{\rm st}\,\eta}{1+(B_{\rm st}-1)\eta}.
    \label{eq:boost_cap}
\end{equation}

The sensing behavior of the random payload is characterized by two statistical moments of the normalized constellation, defined as
\begin{equation}
    \mu_4=\mathbb E\big[|x|^4\big],
    \qquad
    \nu_{-2}=\mathbb E\big[|x|^{-2}\big],
    \label{eq:moments}
\end{equation}
where $\mu_4$ is the fourth-order moment (kurtosis) of the constellation and $\nu_{-2}$ is the inverse second-order moment. As established in \cite{han2026constellation,meng2026constellation,yang2026constellation}, $\mu_4$ governs the payload-induced interference of the MF, and $\nu_{-2}$ governs the noise enhancement of the RF. The inverse moment $\nu_{-2}$ is finite whenever the constellation does not contain the origin. For a constant-modulus constellation such as phase-shift keying (PSK), $\mu_4=\nu_{-2}=1$, whereas any nonconstant-modulus constellation satisfies $\mu_4>1$ and $\nu_{-2}>1$. For compact notation we define the vectorized transmit frame
$\mathbf x=\operatorname{vec}(\mathbf X)\in\mathbb C^{D}$.
\vspace{-0.7\baselineskip}
\subsection{Communication System Model}\label{sec:comm_model}

The communication user experiences a frequency-selective channel with $K_h$ effective taps, $\mathbf g=[g_0,\ldots,g_{K_h-1}]^T\sim\mathcal{CN}(\mathbf 0,\sigma_h^2\mathbf I_{K_h})$, whose delay spread falls within the CP. The user link is modeled as quasi-static over one frame, with the common carrier and Doppler offset of the link handled by receiver synchronization as in standard OFDM. The channel frequency response on subcarrier $n$ is given by
\begin{equation}
    H_n=\sum_{\ell=0}^{K_h-1}g_\ell\,e^{-j2\pi n\ell/N},
    \label{eq:cfr}
\end{equation}
and the received signal on resource element $(n,m)$ is given by
\begin{equation}
    y_{n,m}=H_n\,X_{n,m}+w_{n,m},
    \label{eq:comm_rx}
\end{equation}
where $w_{n,m}\sim\mathcal{CN}(0,\sigma_c^2)$ denotes the additive white Gaussian noise (AWGN) at the communication receiver. The tap variance is normalized to unit average channel gain, $K_h\sigma_h^2=1$, so that $\mathbb E[|H_n|^2]=1$ on every subcarrier, and the per-element communication SNR at unit symbol power is $\mathrm{SNR}_c=1/\sigma_c^2$. Substituting \eqref{eq:cfr} into \eqref{eq:comm_rx} at a pilot cell $(n,m)\in\mathcal P$ gives
\begin{equation}
    y_{n,m}
    =\sum_{\ell=0}^{K_h-1} g_\ell\,X_{n,m}\,e^{-j2\pi n\ell/N}+w_{n,m},
    \label{eq:pilot_obs}
\end{equation}
which is linear in $\mathbf g$. Stacking the $D_p$ pilot observations of \eqref{eq:pilot_obs} gives $\mathbf y_p=\bm\Omega\,\mathbf g+\mathbf w_p$, where the row of $\bm\Omega\in\mathbb C^{D_p\times K_h}$ associated with the pilot cell $(n,m)$ has the entries $[\bm\Omega]_{(n,m),\ell}=X_{n,m}\,e^{-j2\pi n\ell/N}$, and $\mathbf w_p\in\mathbb C^{D_p}$ collects the noise samples of the pilot cells. The receiver estimates $\mathbf g$ via the LMMSE estimator
\begin{equation}
    \hat{\mathbf g}
    =\sigma_h^2\bm\Omega^H
     \big(\sigma_h^2\bm\Omega\bm\Omega^H+\sigma_c^2\mathbf I\big)^{-1}
     \mathbf y_p,
    \label{eq:lmmse}
\end{equation}
with error covariance $\mathbf C_\epsilon=(\sigma_h^{-2}\mathbf I +\sigma_c^{-2}\bm\Omega^H\bm\Omega)^{-1}$ \cite{morelli2001comparison}. The estimation error propagates into the payload demodulation as residual interference, and thereby couples the pilot design to the achievable rate approximation of the payload link. The effective rate approximation of the payload link is $R(\eta,\gamma,\mathcal X)=(1-\eta)\,I_{\mathcal X}\big(\mathrm{SINR}_{\rm eff}(\gamma,\eta)\big)$, where $I_{\mathcal X}$ is the constellation-constrained mutual information and $\mathrm{SINR}_{\rm eff}$ is the effective SINR of the payload link under the channel estimation error. The reference rate $R_{\rm base}$ is the rate at the same pilot density with 16-QAM at uniform power\footnote{Since we focus on the power split between pilots and data payloads, per-subcarrier power allocation is not considered. We refer the readers to \cite{zhang2025optimal} for the details of per-subcarrier power allocation design.}. 
\vspace{-0.7\baselineskip}
\subsection{Sensing System Model and Receive Processing}
\label{sec:sensing_model}

After CP removal and a per-symbol $N$-point FFT, the monostatic sensing receiver, which knows $\mathbf x$, observes the received time-frequency frame
\begin{equation}
    \mathbf{Y}=\sum_{i=1}^{K}\alpha_i\,
    \mathbf{X}\odot\bm{\Phi}(q_i,k_i)+\mathbf{W},
    \label{eq:sensing_rx}
\end{equation}
where the range-Doppler phase matrix is given by
\begin{equation}
    [\bm{\Phi}(q,k)]_{n,m}
    = e^{-j2\pi(\frac{qn}{N}+\frac{km}{M})}.
    \label{eq:steering}
\end{equation}
The sum in \eqref{eq:sensing_rx} runs over the $K=K_t+K_c$ scatterers illuminated by the frame, comprising $K_t$ targets of interest and $K_c$ clutter sources. Unlike the user link, each echo carries its own scatterer-specific Doppler shift, which cannot be removed by a common synchronization and is instead a parameter to be estimated. Here, $\alpha_i$ is the complex amplitude of the $i$-th scatterer, which absorbs the path loss and the radar cross-section, $(q_i,k_i)$ denotes its range-Doppler cell, not necessarily integer-valued, and $\mathbf W\sim\mathcal{CN}(0,\sigma^2\mathbf I)$ represents the AWGN at the sensing receiver. The cell coordinates are the delay and the Doppler shift of the scatterer normalized by the cell sizes, that is, $q_i=\tau_i N\Delta f$ and $k_i=f_{d,i}MT_{\rm o}$. The corresponding range and radial velocity are $r_i=q_i\,c/(2N\Delta f)$ and $v_i=k_i\,c/(2f_cMT_{\rm o})$, where $c$ is the speed of light and $f_c$ is the carrier frequency. All scatterer delays are assumed to fall within the CP and all Doppler shifts to satisfy $f_d\ll\Delta f$, so that the channel acts elementwise on the time-frequency grid as in \eqref{eq:sensing_rx} without inter-carrier interference \cite{keskin2021mimo,xu2025does}. 

With $\mathbf y=\operatorname{vec}(\mathbf Y)$ and $\bm\phi(q,k)=\operatorname{vec}(\bm\Phi(q,k))$, the receiver applies a frame-level receive-filter $\mathbf h\in\mathbb C^D$, normalized to $\mathbf h^H\mathbf x=D$ so that a scatterer's mainlobe response is exactly its amplitude. Two conventional choices of $\mathbf h$ serve as fixed baselines throughout, namely the classical MF and RF \cite{rodriguez2023supervised,meng2026constellation},
\begin{equation}
    \mathbf h_{\rm MF}=\mathbf x,
    \qquad
    \mathbf h_{\rm RF}=\mathbf 1\oslash\mathbf x^{*},
    \label{eq:baselines}
\end{equation}
where the RF meets this normalization identically, and the MF scaling $D/\|\mathbf x\|^2$ concentrates at unity as the frame grows. Both filters are applied elementwise, so that their cost grows only linearly with the frame size $D$, and both are fixed processing rules, since no design freedom remains at the receiver once the frame is transmitted. 
\vspace{-0.7\baselineskip}
\subsection{Problem Formulation}
\label{sec:problem}

We formulate the joint transmit full-frame OFDM signaling and receive-filter design for the considered OFDM-ISAC system. Because the payload is communication data, the transmitter does not choose individual symbols, but it does choose the statistics from which they are drawn and the frame in which they are carried. Three quantities are therefore at its disposal, namely the modulation constellation $\mathcal X$, the pilot placement $\mathcal P$, and the power split $\gamma$. They obey the constraints of the frame model, that is, the $C_4$ symmetry and the minimum spacing $d_{\min}(\mathcal X)\ge d_0$ on the constellation, the latter protecting the demodulation distance, the balance condition \eqref{eq:balanced_class} on the placement, and the boost range \eqref{eq:boost_cap} on the power split. On the receive side, the filter may depend on the frame that was actually sent and is recomputed for every frame. The receive variable is therefore the map $\mathbf h(\cdot)$ rather than a single fixed vector.

The design takes three prescribed inputs, namely the pilot density $\eta$, the per-element communication SNR, $\mathrm{SNR}_c=1/\sigma_c^2$ of \eqref{eq:comm_rx}, and a QoS requirement $\zeta$, which is the fraction of a reference rate that the design must guarantee for the communication link. With these fixed, we pose
\begin{subequations}\label{eq:p1all}
\begin{align}
(\mathcal{P}_1)\,\min_{\mathcal X,\,\mathcal P,\,\gamma,\,\mathbf h(\cdot)}
&\;\mathbb E_{\mathbf x}\big[\underbrace{\lambda\|\mathbf h(\mathbf x)\|^2}_{\text{noise energy}}
+\!\!\underbrace{\sum_{(q,k)\in\mathcal S}\!\!
|\mathbf h(\mathbf x)^H(\mathbf x\odot\bm\phi(q,k))|^2}_{\text{ROI sidelobe residual}}\big]
\label{eq:master}\\
\text{s.t.}\;
&\;\mathcal X\in C_4,\quad \mathcal P\in\mathcal B,\quad
\gamma\in[\eta,\gamma_{\max}],
\label{eq:setcon}\\
&\;\mathbf h(\mathbf x)^H\mathbf x=D\quad\forall\,\mathbf x,
\label{eq:gaincon}\\
&\;R(\eta,\gamma,\mathcal X)\ge\zeta R_{\rm base}.
\label{eq:ratecon}
\end{align}
\end{subequations}
Here, $\mathcal S$ is the notch set, which collects the pairwise relative shifts between the declared ROI cells and the scatterer centers and excludes the zero shift, so that all mainlobes are preserved. The scatterer centers that define $\mathcal S$ are assumed available from the detections of the preceding frames, as is standard in region-of-interest processing \cite{yang2026constellation}. The weight $\lambda>0$ balances the filter noise energy against the residual left at those shifts. The unit-gain constraint \eqref{eq:gaincon} holds for every realization of the payload. The objective in \eqref{eq:master} is the post-filter interference-plus-noise power inside the declared ROIs, the noise energy passed by the filter plus the sidelobe residual left at the notched shifts. The unit-gain constraint \eqref{eq:gaincon} enforces the peak response of a scatterer to equal its own amplitude, so that the residual cannot be reduced by shrinking the mainlobe instead of suppressing the sidelobes. Finally, the rate constraint \eqref{eq:ratecon} keeps the communication service at the prescribed level.
\vspace{-0.7\baselineskip}
\subsection{Optimal Sensing Receiver for a Specific OFDM Frame}
\label{sec:specific}

Although $(\mathcal{P}_1)$ couples four heterogeneous variables, the minimization over the receive filter $\mathbf h(\cdot)$ admits an exact solution. Since both the objective and the unit-gain constraint act on each realization of $\mathbf x$ separately, the expectation in \eqref{eq:master} is minimized by minimizing its integrand at every realization. For any fixed transmit design, the filter design therefore reduces to a per-realization quadratic problem with a single linear constraint, in which $\mathbf x$ is a known constant,
\begin{equation}
\begin{aligned}
    \mathbf{h}_{\rm ROI} &= \arg\min_{\mathbf{h}}\;
    \lambda\|\mathbf{h}\|^2
    +\!\!\sum_{(q,k)\in\mathcal{S}}\!\!
    \big|\mathbf{h}^H(\mathbf{x}\odot\bm{\phi}(q,k))\big|^2 \\[3pt]
    &\quad \text{s.t.}\quad \mathbf{h}^H\mathbf{x}=D,
\end{aligned}
\label{eq:roimmf}
\end{equation}
which generalizes the delay-domain ROI-MMF of \cite{yang2026constellation} to the two-dimensional range-Doppler plane. Its closed-form solution follows from the Woodbury identity. With $\mathbf Z=[\,\mathbf x\odot\bm\phi(q,k)\,]_{(q,k)\in\mathcal S} \in\mathbb C^{D\times|\mathcal S|}$,
\begin{equation}
    \mathbf h_{\rm ROI}=\frac{\beta}{\lambda}\Big(\mathbf x
    -\mathbf Z\big(\lambda\mathbf I+\mathbf Z^H\mathbf Z\big)^{-1}
    \mathbf Z^H\mathbf x\Big),
    \label{eq:woodbury}
\end{equation}
with $\beta$ enforcing $\mathbf h^H\mathbf x=D$, so the filter computation scales with the declared ROI size $|\mathcal S|$ rather than with the frame dimension $D$ \cite{yang2026constellation}. The filter, however, cannot remove the randomness of the payload itself. The notches cancel the average response of the frame, and the random fluctuation of the symbol powers around this average leaks through as residual interference. What the filter leaves behind is therefore set by the payload fluctuation, which only the transmit design can lower, by constellation shaping under the rate constraint.
\vspace{-0.5\baselineskip}
% =====================================================================
\section{Full-Frame OFDM Sensing-Communication Performance}
\label{sec:analysis}
Building on the system model and the closed-form receive filter \eqref{eq:woodbury}, this section derives closed-form expressions for the post-filter interference floor, the noise enhancement, and the $(\mathcal{P}_1)$ objective, together with the closed-form achievable rate of the communication link.
\vspace{-0.7\baselineskip}
\subsection{Floor Analysis of Full-Frame OFDM Sensing}\label{sec:floor}
This subsection characterizes the floor that a target competes against at every cell of the range-Doppler map. The floor collects three components, the deterministic response of the mean frame power, the random fluctuation of the payload, and the noise passed by the filter. Throughout, the frame is vectorized column-wise, so the element $n=n_f+n_t N$ collects the subcarrier index $n_f$ and the OFDM-symbol index $n_t$. The receiver forms the map in two steps. An element-wise time-frequency filtering can be expressed as
\begin{equation}
    r_n=h_n^{*}\,y_n,
    \label{eq:tf_filter}
\end{equation}
and the two-dimensional DFT of the filtered frame then reports each cell of the map,
\begin{equation}
    \chi(q,k)=\frac{1}{D}\,
    \bm\phi^{H}(q,k)\big(\mathbf h^{*}\odot\mathbf y\big)
    =\frac{1}{D}\sum_{n=0}^{D-1}r_n\,e^{j\theta_{(q,k)}(n)},
    \label{eq:cell_stat}
\end{equation}
where $\theta_{(q,k)}(n)=2\pi(q\,n_f/N+k\,n_t/M)$ is the phase of the steering vector in \eqref{eq:steering}, linear in the cell pair $(q,k)$. Under the unit-gain constraint \eqref{eq:gaincon}, the statistic at a scatterer's own cell returns its amplitude plus filtered noise. The map is linear in the echoes of \eqref{eq:sensing_rx}, so the leakage of multiple scatterers superimposes and it suffices to analyze one. Let us consider a scatterer of amplitude $\alpha_c$ at the cell $(q_c,k_c)$. Substituting its echo $\alpha_c\,\mathbf x\odot\bm\phi(q_c,k_c)$ of \eqref{eq:sensing_rx} into \eqref{eq:tf_filter} gives
\begin{equation}
    r_n=\alpha_c\,a_n\,e^{-j\theta_{(q_c,k_c)}(n)}+h_n^{*}w_n,
    \label{eq:tf_filtered}
\end{equation}
where $a_n=h_n^{*}x_n$ weights the contribution of the $n$-th element to the map. Now evaluate \eqref{eq:cell_stat} at a cell $(q_e,k_e)$ displaced from the scatterer by the lag $\Delta=(q_c-q_e,\,k_c-k_e)\neq(0,0)$. Substituting \eqref{eq:tf_filtered} into \eqref{eq:cell_stat}, the linearity of the phase gives $\theta_{(q_c,k_c)}(n)-\theta_{(q_e,k_e)}(n)=\theta_\Delta(n)$, so the statistic depends on the lag alone and is written $\chi(\Delta)$,
\begin{equation}
    \chi(\Delta)
    = \frac{\alpha_c}{D}\sum_{n=0}^{D-1}a_n\,
    e^{-j\theta_\Delta(n)}+\chi_w,
    \label{eq:probe}
\end{equation}
where $\chi_w=\tfrac1D\sum_{n=0}^{D-1} h_n^{*}w_n\,e^{j\theta_{(q_e,k_e)}(n)}$ is the filtered noise. A scatterer leaks into every other cell of the map through the weights $a_n$, and the leakage depends only on the lag between the two cells. The three receivers differ only through $a_n$. Up to their unit-gain constants,
\begin{subequations}\label{eq:gain_fact}
\begin{align}
    a_n^{\rm MF}&=|x_n|^2,\\
    a_n^{\rm RF}&=1,\\
    a_n^{\rm ROI}&=\tilde\beta\,|x_n|^2(1-\xi_n),\label{eq:gain_roi}
\end{align}
\end{subequations}
where $\tilde\beta$ restores the unit mainlobe gain. The MF weight is the symbol power itself, so it fluctuates with the data. The RF weight is constant, so the RF has no data-dependent leakage and pays only a noise enhancement. The ROI-MMF form \eqref{eq:gain_roi} is exact. Writing the closed-form filter \eqref{eq:woodbury} as $\mathbf h_{\rm ROI}=\tfrac{\beta}{\lambda}(\mathbf x-\mathbf Z\mathbf b)$ with the optimal amplitude vector
\begin{equation}
    \mathbf b=(\lambda\mathbf I+\mathbf Z^H\mathbf Z)^{-1}
    \mathbf Z^H\mathbf x,
    \label{eq:notch_b}
\end{equation}
and noting that every column of $\mathbf Z$ is $\mathbf x\odot\bm\phi_i$ with $\bm\phi_i=\bm\phi(\Delta_i)$ and $\theta_i(n)=\theta_{\Delta_i}(n)$ the steering vector and the phase of the $i$-th notched lag $\Delta_i\in\mathcal S$, the correction term factors as
\begin{equation}
    \mathbf Z\mathbf b
    =\sum_{i\in\mathcal S}b_i\,(\mathbf x\odot\bm\phi_i)
    =\mathbf x\odot\sum_{i\in\mathcal S}b_i\,\bm\phi_i,
    \label{eq:had_fact}
\end{equation}
so the filter reads elementwise as $h_n=\tfrac{\beta}{\lambda}\,x_n(1-\xi_n^{*})$ and the gain $a_n=h_n^{*}x_n$ takes the form \eqref{eq:gain_roi}. The data-dependent correction collects one complex exponential per notched lag,
\begin{equation}
    \xi_n=\sum_{i\in\mathcal S}b_i^*e^{j\theta_i(n)}.
    \label{eq:notch_coef}
\end{equation} 
These amplitudes are exactly the values that minimize the filter response at the covered lags. The correction $\xi_n$ is small at the operating point, as quantified in Appendix~C, so the gains of the MF and the ROI-MMF in \eqref{eq:gain_fact} share the frame power $|x_n|^2$ as their leading factor. The analysis below therefore builds on the decomposition of $|x_n|^2$ into a deterministic mean and a random fluctuation,
\begin{equation}
    |x_n|^2 = \underbrace{\bar p_n}_{p_p \text{ or } p_d} + \underbrace{\psi_n}_{0 \text{ (pilot)  or } p_d(|s_n|^2-1) \text{ (data)}}.
    \label{eq:split}
\end{equation}

Here, $s_n$ is the payload symbol of $\mathbf S_d$ carried by the element $n$. The pilots, being constant-modulus, carry no power fluctuation, whereas the payload fluctuation is governed by the variance of the symbol power,
\begin{equation}
    \kappa_2=\mathbb E\big[(|x|^2-1)^2\big]=\mu_4-1,
    \label{eq:kappa2}
\end{equation}
where the second equality uses the unit power $\mathbb E[|x|^2]=1$. The fluctuation of \eqref{eq:split} therefore satisfies
\begin{equation}
    \mathbb E[\psi_n]=0,
    \qquad
    \mathrm{Var}(\psi_n)=p_d^2\,\kappa_2,
    \qquad n\in\mathcal D.
    \label{eq:fstats}
\end{equation}
The statistic $\kappa_2$ vanishes for a constant-modulus constellation and equals $0.32$ for 16-QAM. In the delay domain, it is the single constellation statistic to which the floor of every receiver is proportional \cite{yang2026constellation}. We first characterize the deterministic mean component, and then the random fluctuation together with the filtered noise.

\begin{lemma}\label{lem:mean}
At any integer lag $\Delta=(q,k)\neq(0,0)$, the mean profile of the conventional vertical-comb frame satisfies $\sum_{n=0}^{D-1} \bar p_n e^{-j\theta_\Delta(n)}=0$ unless $k=0$ and $q$ coincides with a comb harmonic, $q_{\rm gl}\in\{\pm N/f_p,\pm 2N/f_p,\ldots\}$, where it equals $(p_p-p_d)D_p$. Consequently, the deterministic grating-lobe response at the comb harmonics is given by
\begin{equation}
    \big|\chi_{\rm gl}\big|^2
    =|\alpha_c|^2\big[(p_p-p_d)\,\eta\big]^2,
    \label{eq:grating}
\end{equation}
and $\mathbb E[\chi(\Delta)]=0$ at every other lag.
\end{lemma}
\begin{IEEEproof}
Since $\bar p_n=p_p$ on $\mathcal P$ and $\bar p_n=p_d$ on $\mathcal D$, the mean profile at lag $\Delta=(q,k)$ splits as
\begin{align}
    \sum_{n=0}^{D-1} \bar p_ne^{-j\theta_\Delta(n)}
    &=p_d\sum_{n_f=0}^{N-1}e^{-j2\pi qn_f/N}
    \sum_{n_t=0}^{M-1}e^{-j2\pi kn_t/M}\nonumber\\
    &+(p_p-p_d)\!\!\sum_{(n,m)\in\mathcal P}\!\!
    e^{-j2\pi(qn/N+km/M)}.
    \label{eq:app_mean}
\end{align}
The first term vanishes for every $\Delta\neq(0,0)$ by the orthogonality of the discrete exponentials. For the vertical comb, the second term vanishes unless $k=0$ and $qf_p=0\pmod N$, that is, at the comb harmonics, where it equals $(p_p-p_d)D_p$. Normalizing by $D$ and squaring gives \eqref{eq:grating} with $D_p/D=\eta$.
\end{IEEEproof}

Lemma~\ref{lem:mean} shows that the uniform-power frame, $\gamma=\eta$, carries a perfectly flat mean profile and hence no deterministic artifact. In contrast, any pilot-power boosting $\gamma>\eta$ on the conventional vertical comb produces grating lobes exactly on the zero-Doppler response, where the range cut and the cell-averaging constant false alarm rate (CFAR) window operate \cite{zhang2024cross,bouziane2026optimized}. Power asymmetry is therefore admissible only if the grating lobes fall outside the processed window or are covered by $\mathcal S$. Lemma~\ref{lem:mean} thus fixes the deterministic term of the design objective and identifies the obstacle that the proposed pilot placement must remove before the boost $\gamma>\eta$ becomes usable.

We next turn to the random component of the floor, which combines the payload fluctuation with the noise passed by the filter, in the following lemma.

\begin{lemma}\label{lem:keff}
At any integer lag $\Delta$ off the comb harmonics, the post-filter cell power of each receiver satisfies
\begin{subequations}\label{eq:keff}
\begin{align}
    \mathbb E\big[|\chi_{\rm MF}(\Delta)|^2\big]
    &=\frac{|\alpha_c|^2\,\keff+\sigma^2\delta_N^{\rm MF}}{D},
    \label{eq:keff_mf}\\
    \mathbb E\big[|\chi_{\rm RF}(\Delta)|^2\big]
    &=\frac{\sigma^2\delta_N^{\rm RF}}{D},
    \label{eq:keff_rf}\\
    \mathbb E\big[|\chi_{\rm ROI}(\Delta)|^2\big]
    &=\frac{[\lambda/(\lambda+D)]^{2}|\alpha_c|^2\,\keff
    +\sigma^2\delta_N^{\rm ROI}}{D},
    \quad \Delta\in\mathcal S,
    \label{eq:keff_roi}
\end{align}
\end{subequations}
where the effective centered second moment is given by
\begin{equation}
    \keff=(1-\eta)\,p_d^2\,\kappa_2,
    \label{eq:keffdef}
\end{equation}
and $\delta_N^{\rm MF}$, $\delta_N^{\rm RF}$, and $\delta_N^{\rm ROI}$ denote the unit-gain noise-enhancement factor $\delta_N=\|\mathbf h\|^2/D$ of the respective receivers, evaluated in closed form in Lemma~\ref{lem:noise}. Expression \eqref{eq:keff_roi} holds on the notched set, and at an uncovered lag $\Delta\notin\mathcal S$ the ROI-MMF floor coincides with the MF floor \eqref{eq:keff_mf} to first order in $|\mathcal S|\keff/D$.
\end{lemma}
\begin{IEEEproof}
Please refer to Appendix~A.
\end{IEEEproof}

The floor expression \eqref{eq:keff} involves the unit-gain noise-enhancement factor $\delta_N$, which quantifies the SNR cost that each filter pays for its interference behavior, and the following lemma evaluates it for the three receivers.

\begin{lemma}\label{lem:noise}
The noise-enhancement factors defined in Lemma~\ref{lem:keff} evaluate in closed form to
\begin{subequations}\label{eq:deltaN}
\begin{align}
    \delta_N^{\rm MF} &= 1,\label{eq:deltaN_mf}\\
    \delta_N^{\rm RF} &= \frac{1}{D}\Big(\frac{D_p}{p_p}
    +\frac{D_d}{p_d}\nu_{-2}\Big),\label{eq:deltaN_rf}\\
    \delta_N^{\rm ROI} &\approx \frac{D}{D-\keff|\mathcal S|},\label{eq:deltaN_roi}
\end{align}
\end{subequations}
where \eqref{eq:deltaN_mf} holds up to the vanishing relative fluctuation of the frame energy and \eqref{eq:deltaN_roi} holds to first order in $|\mathcal S|/D$.
\end{lemma}
\begin{IEEEproof}
Please refer to Appendix~B.
\end{IEEEproof}

The first-order approximation in \eqref{eq:deltaN} holds in the regime $|\mathcal S|\ll D$, consistent with the declared-ROI premise of this work. Under this condition, the additional noise power of the ROI-MMF is $|\mathcal S|/D$-weighted and hence negligible, whereas the RF incurs the full constellation-dependent penalty $\nu_{-2}$, which grows as constellation points approach the origin. Lemmas~\ref{lem:keff} and~\ref{lem:noise} together determine the interference and the noise level of the post-filter floor, and the constellation and the power split act on this floor only through $\keff$, to which the objective decomposition and the operating-point selection below both reduce.
\vspace{-0.7\baselineskip}
\subsection{Closed-Form Objective Decomposition}
\label{sec:decomp}

We now evaluate the $(\mathcal{P}_1)$ objective in closed form. Substituting $\mathbf h_{\rm ROI}$ of \eqref{eq:woodbury} back into \eqref{eq:master} leaves an objective in the transmit variables alone. The following theorem evaluates this objective and exposes where each design variable acts.

\begin{theorem}\label{thm:decomp}
Write $\bar u_n=\mathbb E[h_n^{*}x_n]$ for the mean per-cell reference profile. For the unit-gain ROI-MMF filter, the $(\mathcal{P}_1)$ objective admits the exact decomposition
\begin{equation}
\begin{aligned}
\mathcal J
=\;&\underbrace{\lambda\,\mathbb E\|\mathbf h\|^2}_{\substack{\text{noise
control}\\ \text{(filter)}}}
+\underbrace{\sum_{(q,k)\in\mathcal S}
\Big|\sum_{n=0}^{D-1}\bar u_n\,[\bm\phi(q,k)]_n\Big|^2
}_{\substack{\text{deterministic pilot ambiguity}\\
\text{placement }\mathcal P\text{ and boost }\gamma}}\\[2pt]
&+\underbrace{\keff\;
\Omega_{\mathbf h}(\mathcal S)}_{\substack{\text{payload
fluctuation}\\ \kappa_2\ \text{(constellation) and}\ p_d(\gamma)}},
\end{aligned}
\label{eq:explicit_obj}
\end{equation}
where $\keff\,\Omega_{\mathbf h}(\mathcal S)=\sum_{(q,k)\in\mathcal S}\mathrm{Var}\big(\mathbf h^H(\mathbf x\odot\bm\phi(q,k))\big)\ge0$ is the centered payload term, which vanishes at the constant-modulus endpoint. To first order in $|\mathcal S|/D$, the mean reference $\bar u_n$ reduces to the mean-power profile $\bar p_n$, and the payload weight of the ROI-MMF is
\begin{equation}
    \Omega_{\mathbf h}
    =\Big[\frac{\lambda}{\lambda+D}\Big]^{2}D\,|\mathcal S|,
    \label{eq:omega_vals}
\end{equation}
depending only on the regularization and the notch set.
\end{theorem}
\begin{IEEEproof}
Please refer to Appendix~C.
\end{IEEEproof}

\begin{remark}\label{rem:roles}
Theorem~\ref{thm:decomp} exposes the role of each design variable. The constellation enters only through $\kappa_2$. Its shaping therefore trades mutual information against $\kappa_2$, and the rate constraint \eqref{eq:ratecon} selects the operating point on the boundary of this trade-off. The pilot design enters twice, once through the deterministic ambiguity of the boosted pilot pattern and once through the payload power $p_d=(1-\gamma)/(1-\eta)$ inside $\keff$. The receive filter determines the noise term and the payload weight $\Omega_{\mathbf h}$. Since the ROI-MMF suppresses the covered lags down to the level set by $\lambda$, the payload term is small at the operating point. The sensing-relevant effect of $\keff$ therefore arises from the interference left at the uncovered lags.
\end{remark}

\subsection{Effective Communication SINR and Achievable Rate}
\label{sec:comm_analysis}

To close the design loop, we now derive the achievable rate in closed form. For lattice constant-modulus pilots with $K_h\le N/f_p$, the Gram matrix in \eqref{eq:lmmse} is exactly orthogonal,
\begin{equation}
    \bm\Omega^H\bm\Omega=p_pD_p\,\mathbf I_{K_h},
    \label{eq:gram}
\end{equation}
which renders the residual-error trace closed-form. Treating the estimation error as an additional uncorrelated noise term yields the effective SINR as
\begin{equation}
    \mathrm{SINR}_{\rm eff}(\gamma,\eta)
    = \frac{p_d/\sigma_c^2}{(p_d/\sigma_c^2)\,
    \mathrm{tr}(\mathbf C_\epsilon)+1},
    \label{eq:sinreff}
\end{equation}
where the residual-error trace admits the closed form
\begin{equation}
    \mathrm{tr}(\mathbf C_\epsilon)
    =\frac{K_h\sigma_h^2}{1+\sigma_h^2\,p_pD_p/\sigma_c^2},
    \label{eq:trace}
\end{equation}
which stays below $10^{-3}$ whenever $\gamma\,\mathrm{SNR}_c\ge 10^{3}K_h/D$, as at every operating point considered here, so the Gaussian treatment of the estimation error in \eqref{eq:sinreff} is accurate. The effective rate approximation is given by
\begin{equation}
    R(\eta,\gamma,\mathcal X)
    =(1-\eta)\,I_{\mathcal X}\big(\mathrm{SINR}_{\rm eff}
    (\gamma,\eta)\big),
    \label{eq:rate}
\end{equation}
with $I_{\mathcal X}$ the constellation-constrained mutual information \cite{caire1998bit}. Notably, the pair $(\eta,\gamma)$ appears in the sensing floor through $\keff$ and in the communication rate through $R$, so the sensing and communication performance trade off against each other over the same design variables. Minimizing $\keff$ alone, however, drives the constellation toward the constant-modulus endpoint and the power split toward maximal boosting, both of which erode the communication rate.

% =====================================================================
\vspace{-1\baselineskip}
\section{Joint Constellation and Pilot Design}\label{sec:design}

This section develops the transmit design on the basis of the closed-form results of Section~\ref{sec:analysis}. 
\vspace{-0.7\baselineskip}
\subsection{Balanced Staggered Pilot Placement}
\label{sec:pilot_design}

Lemma~\ref{lem:keff} removes the pilot contribution from the fluctuation floor, and \eqref{eq:keffdef} exposes the payload power $p_d^2$ as a multiplicative floor factor. Shifting energy toward the pilots, $\gamma>\eta$, therefore lowers the floor and, through \eqref{eq:sinreff}, simultaneously improves the channel estimate on which the payload rate depends. By Lemma~\ref{lem:mean}, however, a conventional vertical comb with $\gamma\neq\eta$ places deterministic grating lobes on the zero-Doppler response, so the boost is usable only with a placement that removes them.

To this end, the pilots are placed on a symbol subset $\mathcal T$, and the placement is parameterized by the offset vector $\mathbf o$ as
\begin{equation}
    \mathcal P(\mathbf o)=\{(n,m):\ m\in\mathcal T,\ n = o_m\ (\mathrm{mod}\ f_p')\},
    \label{eq:placement}
\end{equation}
so that each pilot-bearing symbol carries a full-band comb of factor $f_p'<f_p$ and offset $o_m$. Emptying part of the symbol axis is compensated by the denser per-symbol comb, so that the pilot count $D_p$, and hence the density $\eta$, is unchanged. The offset vector $\mathbf o$ is drawn such that the resulting placement belongs to the balanced set $\mathcal B$ of \eqref{eq:balanced_class}, that is, every subcarrier carries an equal number $c_p=D_p/N$ of pilots across the frame. The powers \eqref{eq:pilot_data_power} and the boost bound \eqref{eq:boost_cap} are unchanged, and only the placement is redistributed on the two-dimensional grid. The following proposition establishes the two structural invariances of this construction.

\begin{proposition}\label{prop:balanced}
For any $\gamma\in[\eta,\gamma_{\max}]$ and any placement in $\mathcal B$, (i) the frequency-marginal mean-power profile of the frame is exactly flat, $\sum_m\mathbb E|X_{n,m}|^2=M$ for every $n$, and hence the grating lobes of \eqref{eq:grating} vanish identically on the zero-Doppler response at every lag, and (ii) the pilot Gram matrix of the delay steering stack remains exactly orthogonal, $\bm\Omega^H\bm\Omega=p_pD_p\mathbf I_{K_h}$, so that the closed-form effective SINR \eqref{eq:sinreff} remains exact.
\end{proposition}
\begin{IEEEproof}
Please refer to Appendix~D.
\end{IEEEproof}

Proposition~\ref{prop:balanced} is strictly stronger than the windowing admissibility rule below Lemma~\ref{lem:mean}, since it enlarges the admissible power-asymmetry set to the full interval $\gamma\in[\eta,\gamma_{\max}]$ of \eqref{eq:boost_cap} and thereby makes the entire boost range available to the operating-point selection. Proposition~\ref{prop:balanced} guarantees that this structural gain costs nothing on the communication side, since the channel estimator and the closed-form rate remain those of the conventional comb. The displaced pilot-power fluctuation reappears instead as a two-dimensional mask ambiguity, parameterized by the offset vector $\mathbf o$ that determines the pilot set $\mathcal P$ which is given by,
\begin{equation}
    A_{\mathbf o}(q,k)=\frac{1}{D_p^2}
    \bigg|\sum_{(n,m)\in\mathcal P}
    e^{-j2\pi(\frac{qn}{N}+\frac{km}{M})}
    -\eta D\,\delta_{q,0}\delta_{k,0}\bigg|^2,
    \label{eq:mask_amb}
\end{equation}
which is nonzero only at $k\neq 0$ and is itself a design degree of freedom, exercised in two steps. First, the pilot-bearing symbol positions are chosen aperiodically. A periodic symbol subset concentrates \eqref{eq:mask_amb} into deterministic spikes at the Doppler harmonics of its period, whereas an aperiodic subset spreads the same energy across the Doppler axis. Second, for the chosen positions, the offset permutation minimizes the weighted cost
\begin{equation}
    \min_{\mathbf o}\;\Big\{
    \max_{(q,k)}A_{\mathbf o}
    +w_t\,A_{\mathbf o}^{\text{top-}J}
    +w_r\!\!\max_{(q,k)\in\mathrm{ROI}}\!\!A_{\mathbf o}
    +w_m\,\bar A_{\mathbf o}^{\mathrm{ROI}}\Big\},
    \label{eq:mask_opt}
\end{equation}
where $A_{\mathbf o}^{\text{top-}J}$ is the mean of the $J$ largest values of $A_{\mathbf o}$ outside the origin, $\bar A_{\mathbf o}^{\mathrm{ROI}}$ is its mean over the declared ROIs, and $w_t,w_r,w_m>0$ are fixed weights prioritizing the declared range-Doppler ROIs. Since the balanced placements form a permutation class, \eqref{eq:mask_opt} is solved by randomized restarts and balance-preserving pairwise offset swaps. When $p_p\neq p_d$, the nonzero-Doppler harmonic lags of the pilot lattice are additionally included in $\mathcal S$, so that no deterministic component of the mask ambiguity survives inside the processed region.

\vspace{-1.0\baselineskip}
\subsection{Constellation Shaping and Pilot-Data Power Allocation}
\label{sec:gcs}

The remaining variables are the constellation $\mathcal X$ and the power split $\gamma$. They enter the floor only through $\keff$ in \eqref{eq:keffdef} and enter the rate through \eqref{eq:rate}, so tightening the constellation toward constant modulus and boosting the pilots both lower the floor and both reduce the payload rate. The rate constraint \eqref{eq:ratecon} resolves this trade-off, and the pair $(\mathcal X^\star,\gamma^\star)$ is selected jointly. The selection proceeds in two stages. A family of $C_4$-symmetric geometric constellation shaping (GCS) sets, ordered in $\kappa_2$, is first constructed. The operating pair is then selected on this family by a joint optimization on the rate boundary.

We parameterize an $L$-point constellation by $L/4$ free complex points replicated under $90^\circ$ rotations, $\mathcal X(\mathbf z)=\{j^rz_u:\ r\in\{0,1,2,3\},\ u=1,\ldots,L/4\}$, normalized to unit average power. The $C_4$ symmetry guarantees an exactly zero mean and I/Q balance for every iterate, so that \eqref{eq:c4_moments} holds by construction. The design objective is the exact finite-alphabet mutual information of the equiprobable AWGN channel $Y=x+n$, $n\sim\mathcal{CN}(0,1/s)$ with $s$ the evaluation SNR,
\begin{equation}
    I_{\mathcal X}(s)=\log_2 L-\frac{1}{L}\sum_{i=1}^{L}
    \mathbb E_{n}\!\left[\log_2\sum_{j=1}^{L}
    e^{-s\left(|x_i-x_j+n|^2-|n|^2\right)}\right],
    \label{eq:gh_mi}
\end{equation}
evaluated by Monte Carlo sampling during training and by Gauss-Hermite quadrature for the rate constraint \eqref{eq:ratecon}.

For each receive filter, we solve the receiver-matched moment-constrained problem
\begin{equation}
    \max_{\mathcal X\in C_4}\;
    I_{\mathcal X}(\mathrm{SNR}_c^{\rm des})
    \quad\text{s.t.}\quad
    \Gamma(\mathcal X)\le\tau_\rho,
    \label{eq:gcs_opt}
\end{equation}
where the sensing statistic $\Gamma$ is matched to the deployed receiver following Lemmas~\ref{lem:keff} and~\ref{lem:noise}. Specifically, we set $\Gamma=\mu_4$, equivalently $\kappa_2=\mu_4-1$, for the ROI-MMF, whose floor statistic the MF shares, and $\Gamma=\nu_{-2}$ for the RF. The constraint level $\tau_\rho$ is swept over seven values between the unconstrained communication-optimal endpoint, $\rho=0$, and the constant-modulus endpoint, $\rho=6$, which corresponds to $L$-PSK. Since \eqref{eq:gcs_opt} is nonconvex, to avoid poor local optima, each level is solved by penalized gradient ascent with multiple restarts from QAM, concentric ring, and random initializations, as well as warm-started from the previous and looser level and the run with the largest mutual information among those that satisfy the constraint is retained. Traversing the levels by using this method keeps the constellations nested, an ordering used by the operating-point selection.

The receiver-matched constellation sets are then re-gridded onto a common set of seven equal-MI targets $\{I_\rho\}$ by exchanging the roles of the objective and the constraint,
\begin{equation}
    \min_{\mathcal X\in C_4}\;\Gamma(\mathcal X)
    \quad\text{s.t.}\quad
    I_{\mathcal X}(\mathrm{SNR}_c^{\rm des})\ge I_\rho,
    \label{eq:gcs_mi}
\end{equation}
solved by the same penalized multi-restart ascent. The resulting levels are matched across the receiver-specific sets, as summarized in Table~\ref{tab:family}, so that every receiver comparison that follows is iso-rate at every level. The constellation sets are constructed once at $\mathrm{SNR}_c^{\rm des}=10$~dB and fixed thereafter, and the operating point among them is selected below at the declared operating SNR.

\begin{table}[t!]
\centering
\caption{Equal-MI $L{=}16$ shaped constellation levels at the design
SNR of $10$~dB.}
\label{tab:family}
\begin{tabular}{lccccccc}
\hline
$\rho$       & 0     & 1     & 2     & 3     & 4     & 5     & 6     \\
\hline
MI (bit)     & 3.200 & 3.134 & 3.054 & 2.977 & 2.903 & 2.828 & 2.746 \\
$\kappa_2$   & 0.402 & 0.190 & 0.128 & 0.092 & 0.060 & 0.033 & 0.000 \\
$\nu_{-2}$   & 2.040 & 1.339 & 1.207 & 1.134 & 1.080 & 1.038 & 1.000 \\
\hline
\end{tabular}
\end{table}

A communication SNR, $\mathrm{SNR}_c$ and the QoS requirement $\zeta$ are prescribed, and the reference rate $R_{\rm base}$ in \eqref{eq:ratecon} is instantiated as the rate of the 16-QAM frame at uniform power and the reference density,
\begin{equation}
    R_{\rm base}
    =(1-\eta)\,I_{16\text{-}\mathrm{QAM}}
    \big(\mathrm{SINR}_{\rm eff}(\eta,\eta)\big).
    \label{eq:rbase}
\end{equation}
With the grating lobes removed by Proposition~\ref{prop:balanced}, the operating-point selection extends from $\mathcal X$ alone to the pair $(\mathcal X,\gamma)$,
\begin{align}
    (\mathcal X^\star,\gamma^\star)
    &=\arg\min_{\mathcal X\in C_4,\; \gamma\in[\eta,\gamma_{\max}]}\; \keff=(1-\eta)\,p_d^2\,\kappa_2(\mathcal X) \label{eq:joint_opt}\\
    \text{s.t.}\quad &(1-\eta)\,I_{\mathcal X}\big(\mathrm{SINR}_{\rm eff}(\gamma,\eta)\big)\ge\zeta R_{\rm base}, \nonumber \\
    &d_{\min}(\mathcal X)\ge d_0, \nonumber
\end{align}
where $d_0$ is a prescribed minimum symbol spacing.

Problem \eqref{eq:joint_opt} is solved by examining the shaping levels of Table~\ref{tab:family} one at a time. For each level, the boost $\gamma$ is found by a bisection on the rate constraint \eqref{eq:ratecon}. For every level with $\kappa_2>0$, the floor decreases strictly with the boost through $p_d^2$, so the optimum $\gamma^\star$ is simply the largest boost that satisfies \eqref{eq:ratecon}. At every evaluated operating point this boost lies strictly below the upper bound $\gamma_{\max}$, so the rate constraint is active at the optimum.

Among the shaping levels of Table~\ref{tab:family}, the design selects the tightest level whose boundary rate still satisfies \eqref{eq:ratecon}. Since the levels are nested and ordered in both rate and $\kappa_2$, this level attains the smallest feasible $\keff$. Finally, we remark on the role of the pilot density. The density $\eta$ is a prescribed parameter fixed by the standard, and the QoS baseline is defined at the same density. The design is solved once per density and is density-robust, since the $(1-\eta)$ prefactor common to $\keff$ and to the baseline floor cancels in their ratio and the boundary payload powers are nearly equal, so that every $\eta$ yields the same shaping level and floor reduction.

\begin{comment}
\begin{table}[t!]
\centering
\caption{Three-step design procedure for $(\mathcal{P}_1)$.}
\label{tab:procedure}
\begin{tabular}{L{1.05cm}L{1.25cm}L{5.2cm}}
\hline
Step & Variable & Mechanism and result \\
\hline
$\mathbf h$-step & $\mathbf h(\cdot)$ &
Closed form per frame: the ROI-MMF \eqref{eq:roimmf} via the Woodbury
form \eqref{eq:woodbury}, with cost scaling in $|\mathcal S|$. \\
$\mathcal P$-step & $\mathcal P$ &
Structural: balanced staggering removes the deterministic term of
\eqref{eq:explicit_obj} while preserving the channel-estimation Gram,
and the residual ambiguity is minimized by \eqref{eq:mask_opt}. \\
$(\kappa_2,\gamma)$-step & $\mathcal X,\gamma$ &
Discrete selection on the rate boundary: \eqref{eq:joint_opt} solved by
selecting the tightest level that satisfies \eqref{eq:ratecon}. \\
\hline
\end{tabular}
\end{table}
The complete design is compiled in Table~\ref{tab:procedure}.
\end{comment}

The design is thus obtained in three steps. The receive filter follows in closed form from \eqref{eq:woodbury} for each transmitted frame, the placement is fixed by the balanced staggered construction of \eqref{eq:mask_opt}, and the constellation and the power split are selected by \eqref{eq:joint_opt}. We now show in what sense this sequential procedure recovers the joint optimum of $(\mathcal P_1)$.

\begin{proposition}\label{prop:opt}
The three-step procedure attains the first-order optimum of $(\mathcal{P}_1)$ over the shaping levels of Table~\ref{tab:family}, up to the residual harmonic-lag ambiguity of the balanced pilot placement. The neglected terms are smaller than the retained ones by the factor $|\mathcal S|/D$.
\end{proposition}
\begin{IEEEproof}
Please refer to Appendix~E.
\end{IEEEproof}

\begin{remark}\label{rem:procedure}
The proposed procedure addresses $(\mathcal{P}_1)$ in a single pass. The filter step is exact and the placement step is structural, whereas the selection of $(\mathcal X,\gamma)$ is a discrete search, which the strict monotonicity of the objective in $\keff$ established in Theorem~\ref{thm:decomp} reduces to the choice of a single shaping level.
\end{remark}
\vspace{-1\baselineskip}
% =====================================================================

\section{Numerical Results}\label{sec:results}

In this section, we validate the theoretical framework and the proposed joint design through comprehensive numerical simulations. Unless stated otherwise, we consider $N=256$ subcarriers, $M=32$ OFDM symbols, and a bandwidth of $50$~MHz corresponding to a range-cell size of $3$~m, with a regularization parameter $\lambda=0.1$, a declared $\mathrm{SNR}_c=15$~dB at reference pilot density $\eta=1/8$, a carrier frequency of $f_c=28$~GHz with a CP of $T_{\rm cp}=T_{\rm u}/8$, which gives a velocity-cell size of $29.1$~m/s, and $K_h=8$ channel taps. The evaluated scene places an on-grid clutter return at cell $(0,0)$ at $40$~dB above the noise floor, target~1 at $(-24,0)$, and target~2 at $(18,8)$, both at a target SNR of $10$~dB. The declared ROIs are the boxes of half-widths $(6,3)$ and $(8,3)$ around the two targets, and the notch set $\mathcal S$ is the pairwise relative-shift construction defined below \eqref{eq:p1all}. Each Monte Carlo point is averaged over $1000$ to $3000$ independent trials. The placement cost \eqref{eq:mask_opt} uses the weights $(w_t,w_r,w_m)=(0.25,0.60,0.15)$ with $J=24$, minimized over $600$ random balanced initializations of which the eight best are refined by $500$ pairwise offset swaps each, and the best four are retained. All quoted floor and notch levels are trial-averaged mean powers.

The MF, RF, and ROI-MMF baselines are evaluated with the transmitted signal under fixed 16-QAM at uniform power, on the lattice pilot frame with $\gamma=\eta$ for the map, cut, and detection figures, and on an all-payload frame for the estimation benchmark. The proposed design operates at the selected shaping level $\kappa_2=0.128$ with the boost $\gamma^\star$ on the rate boundary and is labeled \emph{Full joint} in the figures. All rates are reported through the same effective SINR \eqref{eq:sinreff} at the declared $\mathrm{SNR}_c$.
\vspace{-1.0\baselineskip}
\subsection{Joint Constellation and Pilot Design}\label{sec:joint_results}

Fig.~\ref{fig:GCS} shows the $L\in\{16,64\}$ constellation sets obtained by maximizing the mutual information under the receiver-matched moment constraint in \eqref{eq:gcs_opt}, swept from the sensing-optimal constant-modulus endpoint to the communication-optimal endpoint. For each order, the $\mu_4$-constrained set is shared by the MF and the ROI-MMF, the $\nu_{-2}$-constrained set serves the RF, and each panel is annotated with the equal-MI target and the corresponding moment at the design SNR. At the constant-modulus endpoint, the two constellation sets of each order coincide with $L$-PSK, for which $\mu_4 = 1 (\kappa_2 = 0)$ and $\nu_{-2}=1$. Toward the communication-optimal endpoint the $\mu_4$-constrained set reduces the power dispersion at nearly uniform ring occupancy, whereas the $\nu_{-2}$-constrained set additionally repels points from the origin, consistent with the receiver-specific statistics of Lemmas~\ref{lem:keff} and~\ref{lem:noise}. The equal-MI guarantees that, at every level, the constellation sets of each order operate at identical mutual information.

\begin{figure}[t!]
    \centering
    {\includegraphics[width=0.45\textwidth]{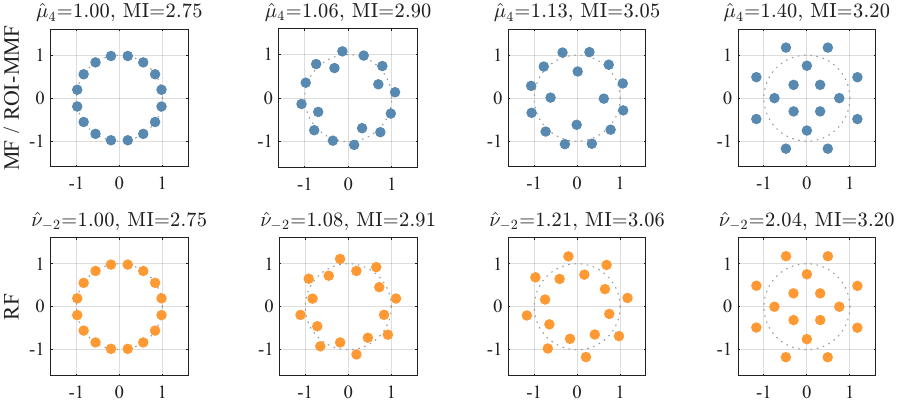}}\\
    {\small (a)}\\
    {\includegraphics[width=0.45\textwidth]{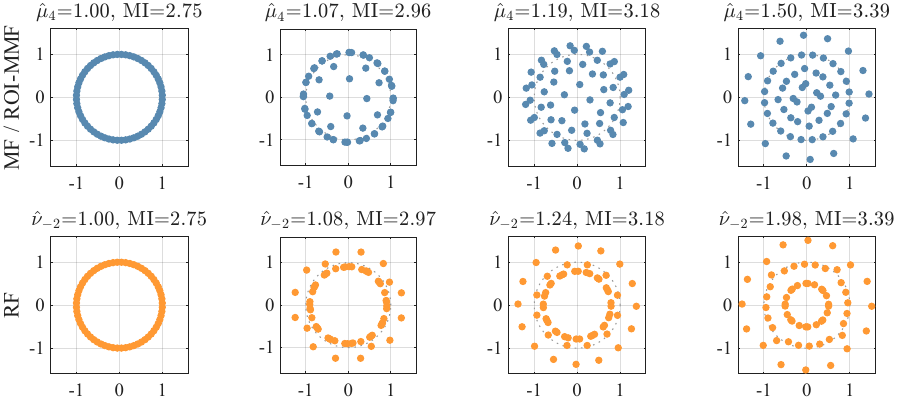}}\\
    {\small (b)}\\
    \caption{Receiver-specific $C_4$-symmetric GCS constellation sets
    for (a) $L=16$ and (b) $L=64$, swept over the seven
    equal-MI shaping levels at the design SNR of $10$~dB.}
    \label{fig:GCS}
\end{figure}

\begin{figure}[t!]
    \centering
    \setlength{\tabcolsep}{1pt}
    \begin{tabular}{cc}
    \includegraphics[width=0.49\columnwidth]{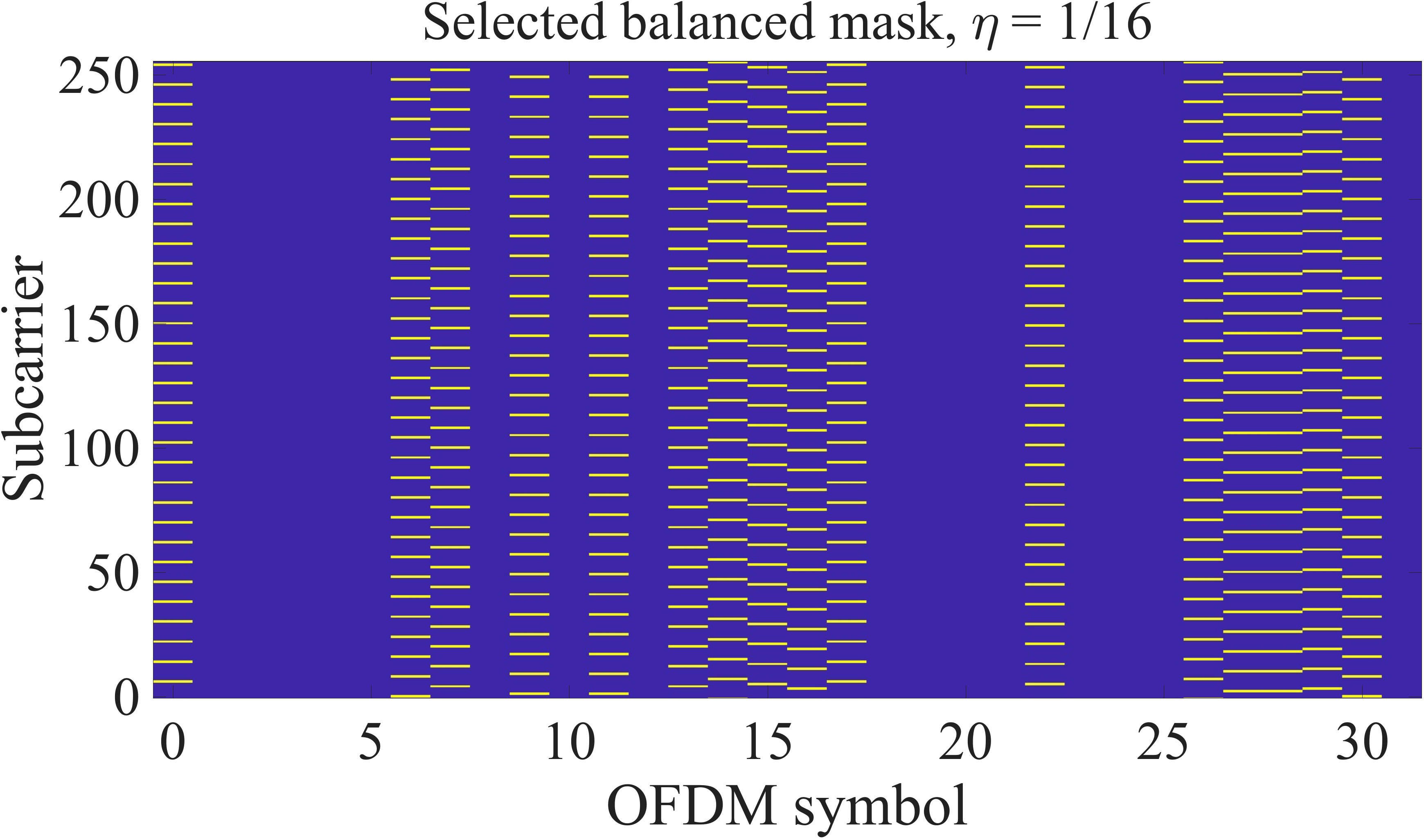} &
    \includegraphics[width=0.49\columnwidth]{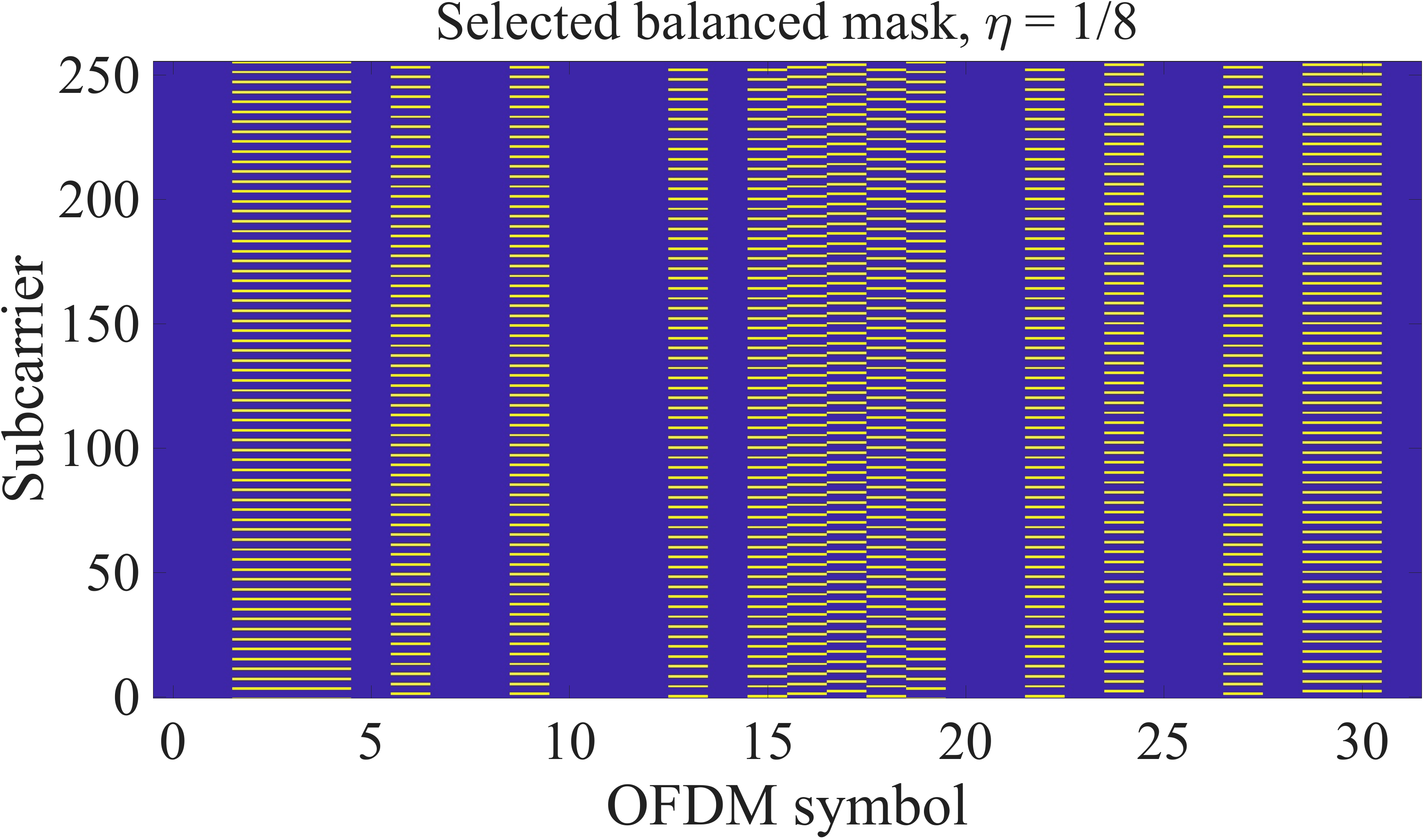} \\
    {\small (a)} & {\small (b)} \\
    \includegraphics[width=0.49\columnwidth]{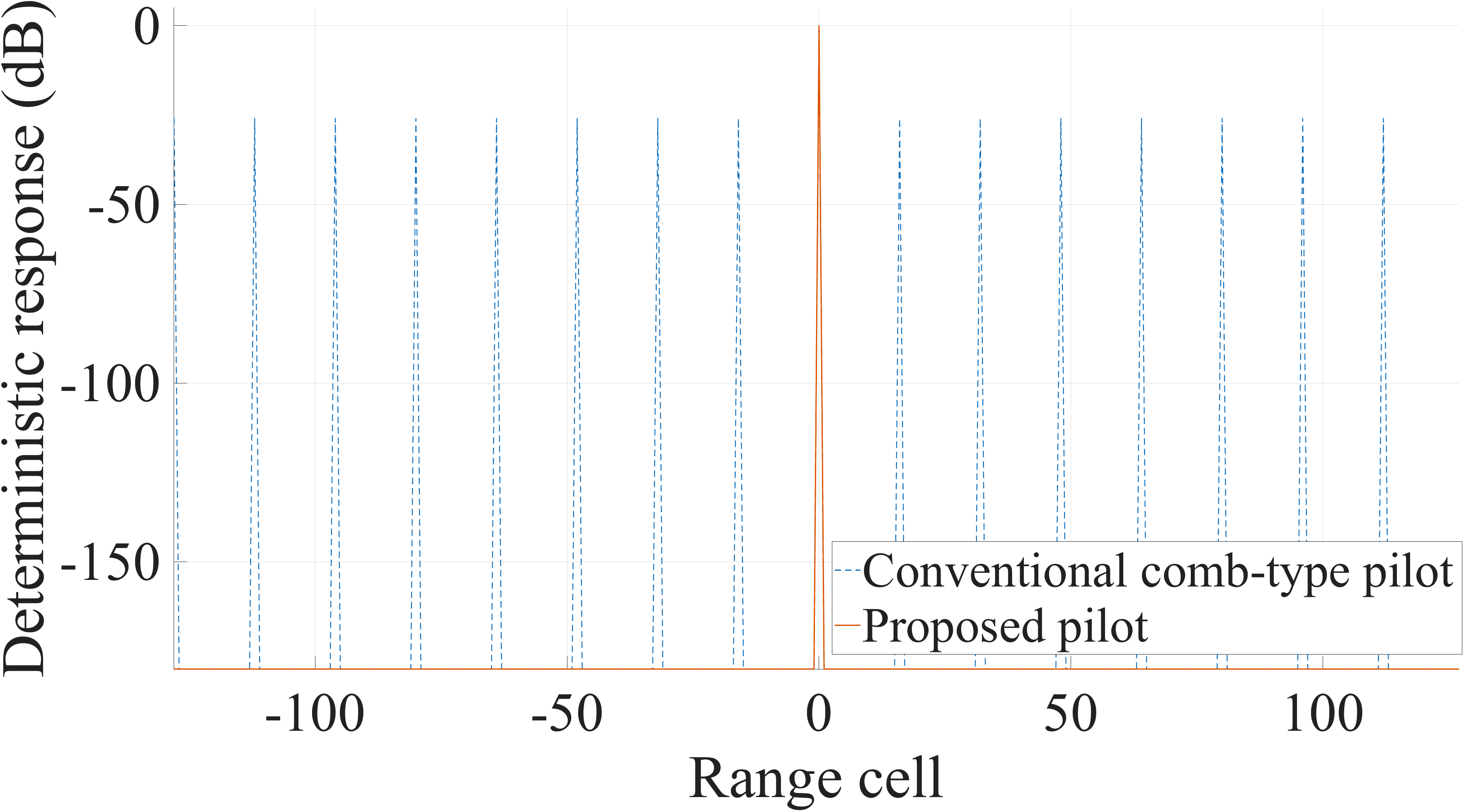} &
    \includegraphics[width=0.49\columnwidth]{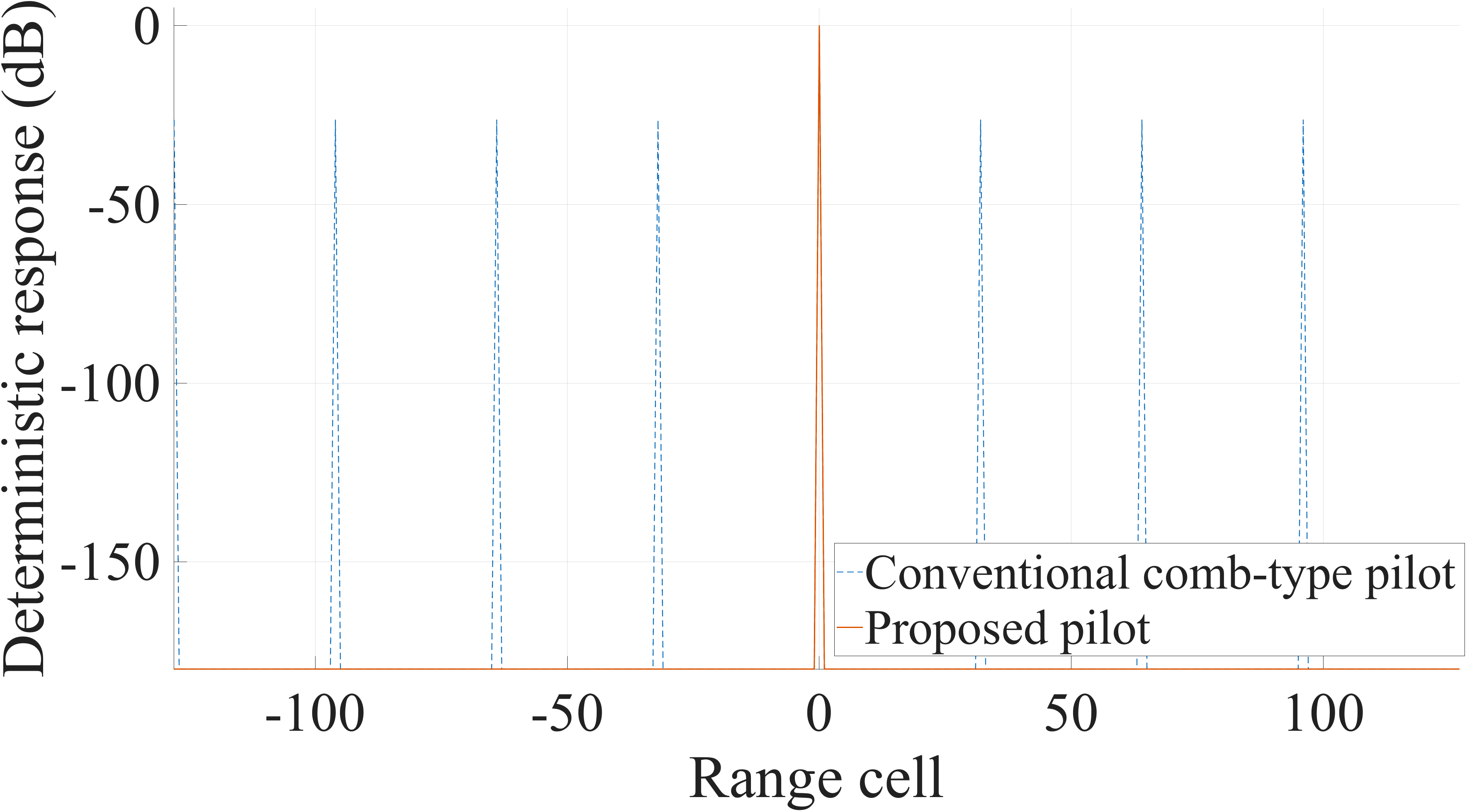} \\
    {\small (c)} & {\small (d)} \\
    \end{tabular}
    \caption{Pilot-design diagnostics for $\eta=1/16$ and $1/8$ at the
    operating boosts $\gamma^\star$:
    (a) and (b) the balanced staggered masks, and (c) and (d) the
    zero-Doppler responses of the boosted conventional vertical comb and
    the balanced placement.}
    \label{fig:pilot_diag}
\end{figure}

Fig.~\ref{fig:pilot_diag} verifies the two halves of Proposition~\ref{prop:balanced} at the operating point. Despite the $\gamma^\star>\eta$ boost, the balanced placement leaves no deterministic grating lobes on the zero-Doppler response, whereas the conventional vertical comb under the same boost exhibits grating lobes exactly at the harmonic lags $\pm N/f_p,\pm2N/f_p,\ldots$ with the closed-form level \eqref{eq:grating}. Furthermore, the optimized offset ordering confines the displaced pilot ambiguity \eqref{eq:mask_amb} to a few nonzero-Doppler alias lines outside the declared ROIs.
\vspace{-0.7\baselineskip}
\subsection{Sensing Performance of Full-Frame OFDM-ISAC}\label{sec:sensing_results}

Fig.~\ref{fig:cuts} examines the post-filter floor predicted by \eqref{eq:keff} on the zero-Doppler and Doppler cuts, comparing the 16-QAM and the equal-MI shaped MF and RF baselines at $\gamma=\eta$ with the proposed joint design at its fixed operating point. The MF outputs retain a high floor over the entire cut, under both the 16-QAM and the shaped constellation, and the weak targets do not stand clearly above. The RF outputs are flat near the noise level and carry the $\nu_{-2}$ enhancement of Lemma~\ref{lem:noise} in place of a payload floor. In contrast, the proposed joint design suppresses the declared regions down to the noise level, so that the weak targets become clearly visible there, and its uncovered floor lies below the MF floor at the declared QoS level $\zeta=0.98$, as the reduction of $\keff$ predicts. Moreover, the boost $\gamma^\star>\eta$ produces no grating lobes, which confirms that the balanced placement removes the deterministic grating response, as established in Proposition~\ref{prop:balanced}.

Fig.~\ref{fig:rdmap} repeats the comparison on the two-dimensional map, where masking is decided cell by cell. Under the MF the response at both target cells is $-44.5$~dB, which is the level of the $\keff$ floor rather than of the targets, so that neither target is separable from its surroundings. Under the proposed joint design both targets stand $+49$~dB above the notch level and are clearly detectable. The uncovered floor also drops by $4.3$~dB relative to the 16-QAM ROI-MMF, in close agreement with the $4.4$~dB reduction of $\keff$ predicted at the declared QoS level.

\begin{figure}[t!]
    \centering
    \begin{minipage}[b]{0.49\columnwidth}
        \centering
        \includegraphics[width=\linewidth]{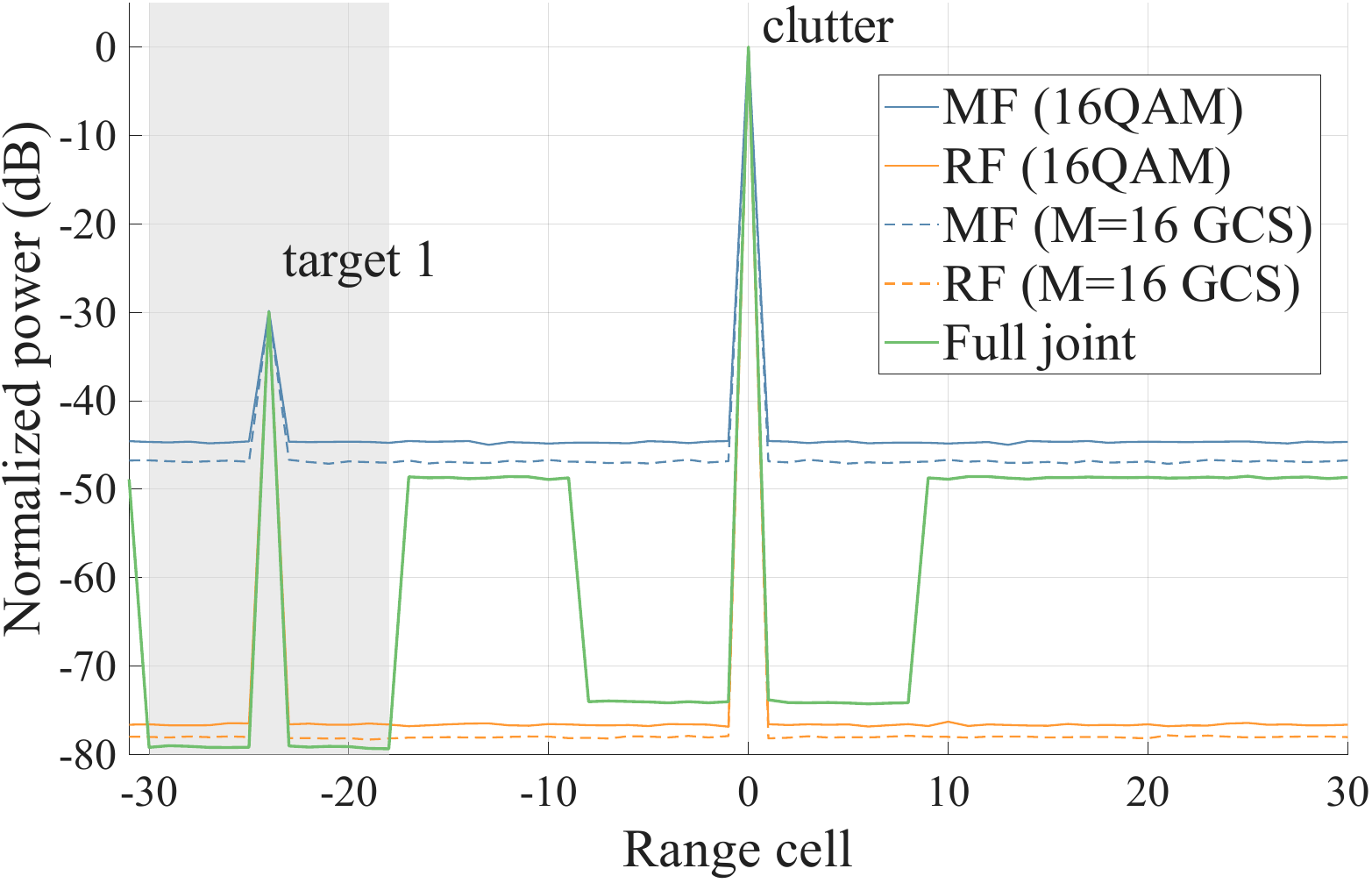}\\
    \end{minipage}\hfill
    \begin{minipage}[b]{0.49\columnwidth}
        \centering
        \includegraphics[width=\linewidth]{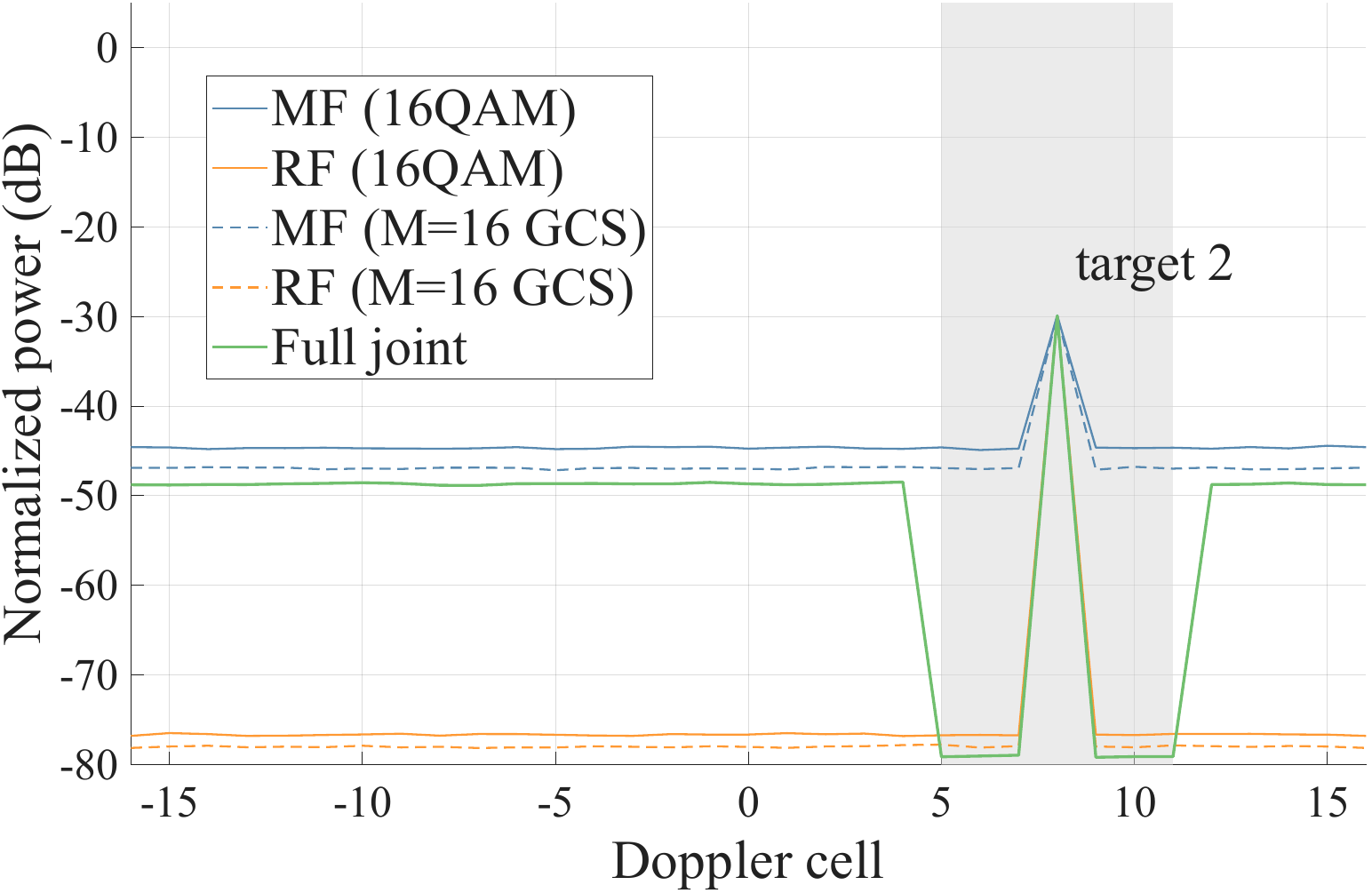}\\
    \end{minipage}

\caption{Range-Doppler scene cuts at $\eta=1/8$
($\zeta=0.98$, clutter at $40$~dB, target SNR $10$~dB). Left shows the zero-Doppler range cut through the clutter and target~1, and right shows the Doppler cut at the target-2 range.}
    \label{fig:cuts}
\end{figure}

\begin{figure}[t!]
    \centering
    \includegraphics[width=\linewidth]{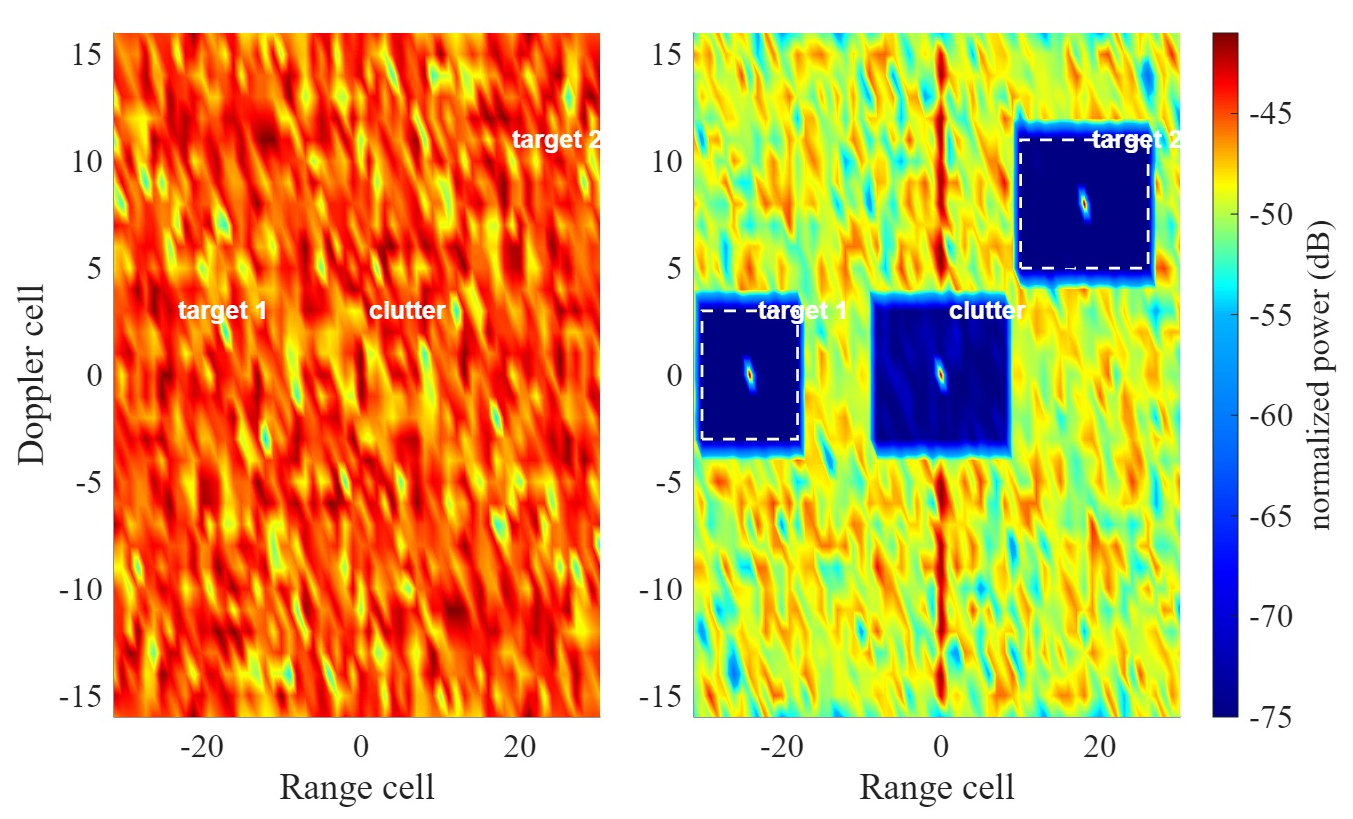}
    \caption{Range-Doppler maps. Left is the result of 16-QAM MF baseline. Right is the result of
    proposed joint design ($\eta=1/8$, $\gamma^\star=0.167$,
    $\zeta=0.98$). Dashed boxes mark the declared ROIs.}
    \label{fig:rdmap}
\end{figure}

Fig.~\ref{fig:rmse}(a) examines the delay-estimation accuracy of the MF and the RF, evaluated on both the fixed 16-QAM and the equal-MI shaped geometries, and of the proposed joint design, where the range estimates are obtained by the matrix-pencil estimator \cite{sarkar1995using}, as in \cite{yang2026constellation}. The dotted lines are the root CRB (RCRB). Applying the standard Fisher information analysis yields the per-target bounds
\begin{equation}
\begin{aligned}
    \mathrm{RCRB}_{\hat r_i}
    &=\frac{c}{2N\Delta f}
    \sqrt{\frac{\sigma^{2}N^{2}}
    {8\pi^{2}|\alpha_i|^{2}\sum_{n_f=-N/2}^{N/2-1}n_f^{2}}},\\
    \mathrm{RCRB}_{\hat v_i}
    &=\frac{c}{2f_cMT_{\rm o}}
    \sqrt{\frac{\sigma^{2}M^{2}}
    {8\pi^{2}N|\alpha_i|^{2}\sum_{n_t=-M/2}^{M/2-1}n_t^{2}}},
\end{aligned}
    \label{eq:crb}
\end{equation}
in meters and meters per second \cite{han2026constellation}, where the leading factors are the range-cell and the velocity-cell sizes of Section~\ref{sec:sensing_model}, the subcarrier and symbol indices run symmetrically about the band center and the frame center, and the factor $N$ in the velocity bound is the coherent gain of the range matching. The MF curves flatten for both the 16-QAM and the shaped payload, and the two floors differ by the amount by which shaping reduces $\keff$. The RF curve remains parallel to the RCRB over the whole sweep, consistent with the constant multiplicative gap $\sqrt{\nu_{-2}}$. Finally, the proposed joint design follows the bound most closely and reaches $0.025$~m against a bound of $0.023$~m at a target SNR of $+10$~dB. Fig.~\ref{fig:rmse}(b) repeats the comparison on the Doppler axis, where the same estimator is applied to the slow-time sequence at each target range. The same ordering is observed, and the proposed joint design again follows the bound most closely, reaching $0.042$~m/s against a bound of $0.040$~m/s at the same target SNR.

Fig.~\ref{fig:pd} examines the weak-target detection probability under CFAR processing at $P_{\rm FA}=10^{-4}$ inside the declared ROIs, which is governed by the post-filter floor of Lemma~\ref{lem:keff} and therefore by $\keff$. By \eqref{eq:keff_mf}, the covered clutter inflates the CFAR training cells by the floor excess
\begin{equation}
    \Delta_{\rm ex}
    =10\log_{10}\Big(1+\frac{|\alpha_c|^2\keff}{\sigma^2}\Big).
    \label{eq:excess}
\end{equation}
The measured excess is $+34.5$~dB for the 16-QAM MF, in agreement with \eqref{eq:excess}, and $+32.2$~dB for the shaped payload. The RF excess stays below $+2.5$~dB and reflects its $\nu_{-2}$ penalty instead, while the proposed design shows essentially zero excess because the covered interference is suppressed inside the region. Consequently, the proposed design achieves the smallest detection threshold, with the $P_D=0.9$ thresholds ranging from $-27$~dB for the proposed design to $+7$~dB for the MF, a separation of $34$~dB, consistent with the $\keff$ ordering of Lemma~\ref{lem:keff}.
\vspace{-0.7\baselineskip}
\subsection{Sensing-Communication Trade-Off}\label{sec:tradeoff_results}

\begin{figure}[t!]
    \centering
    \begin{minipage}{0.49\linewidth}\centering
        \includegraphics[width=\linewidth]{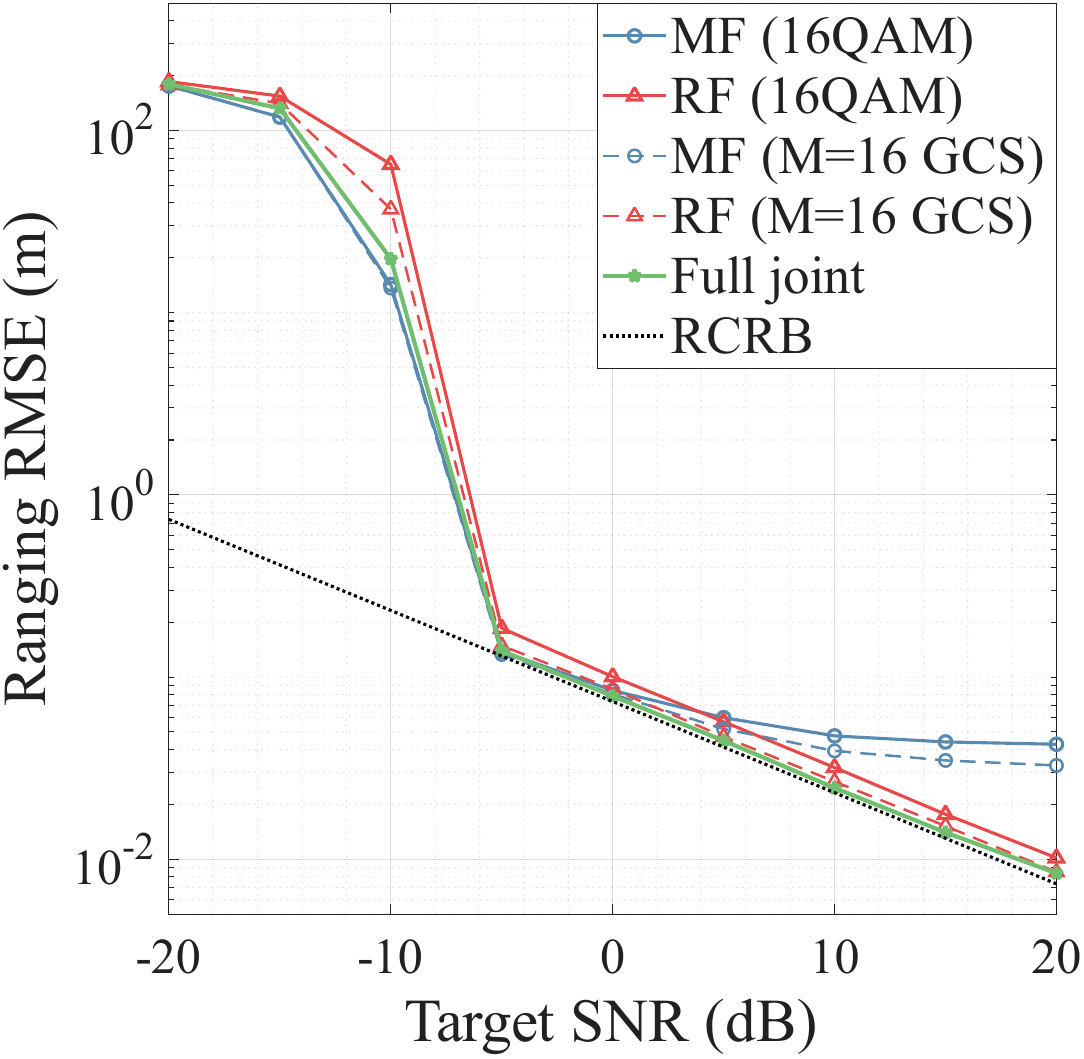}\\[2pt]
        {\footnotesize (a)}
    \end{minipage}\hfill
    \begin{minipage}{0.49\linewidth}\centering
        \includegraphics[width=\linewidth]{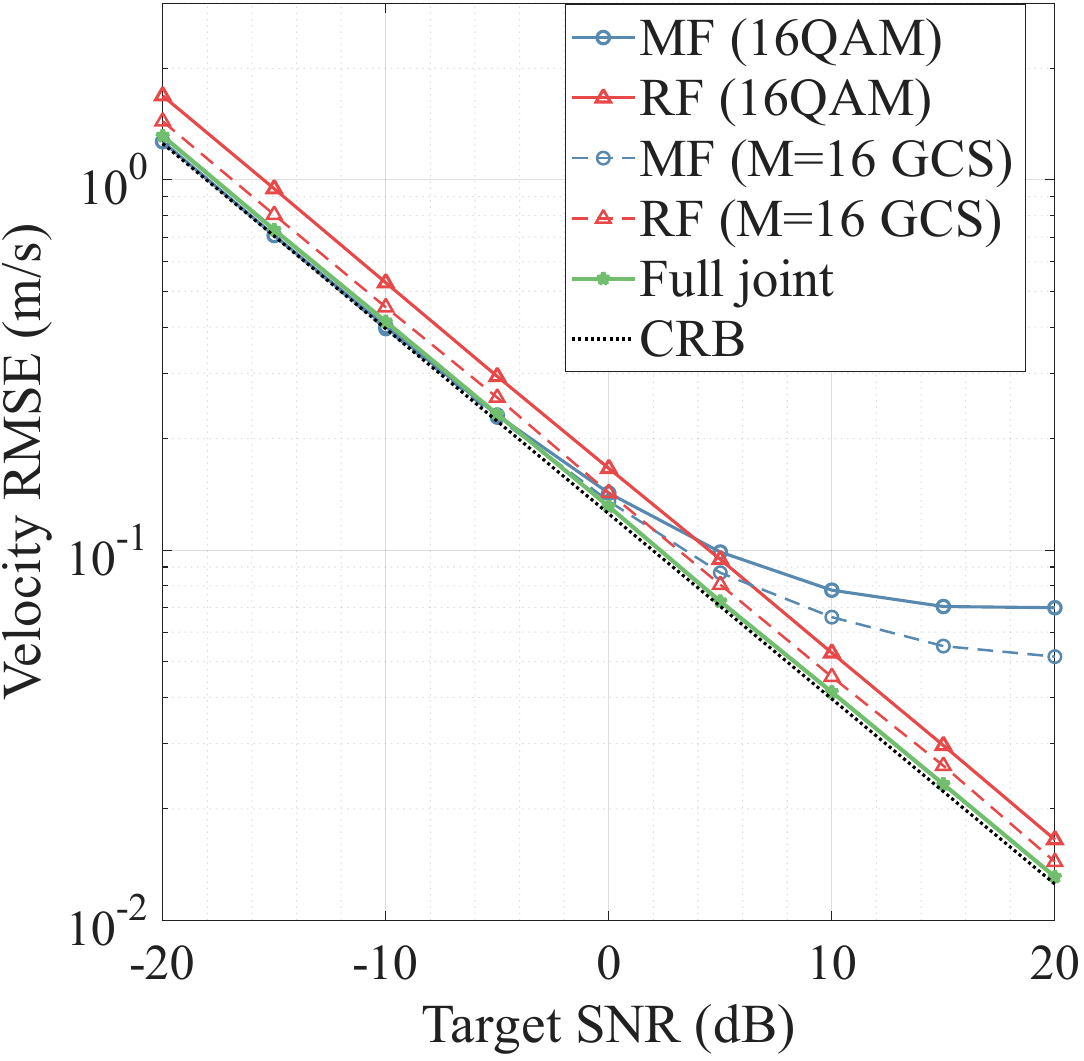}\\[2pt]
        {\footnotesize (b)}
    \end{minipage}
    \caption{(a) Ranging RMSE and (b) velocity RMSE versus target SNR
    for the two targets inside the declared ROIs ($\eta=1/8$,
    $\zeta=0.98$).}
    \label{fig:rmse}
\end{figure}

\begin{figure}[t!]
    \centering
    \includegraphics[width=0.95\linewidth]{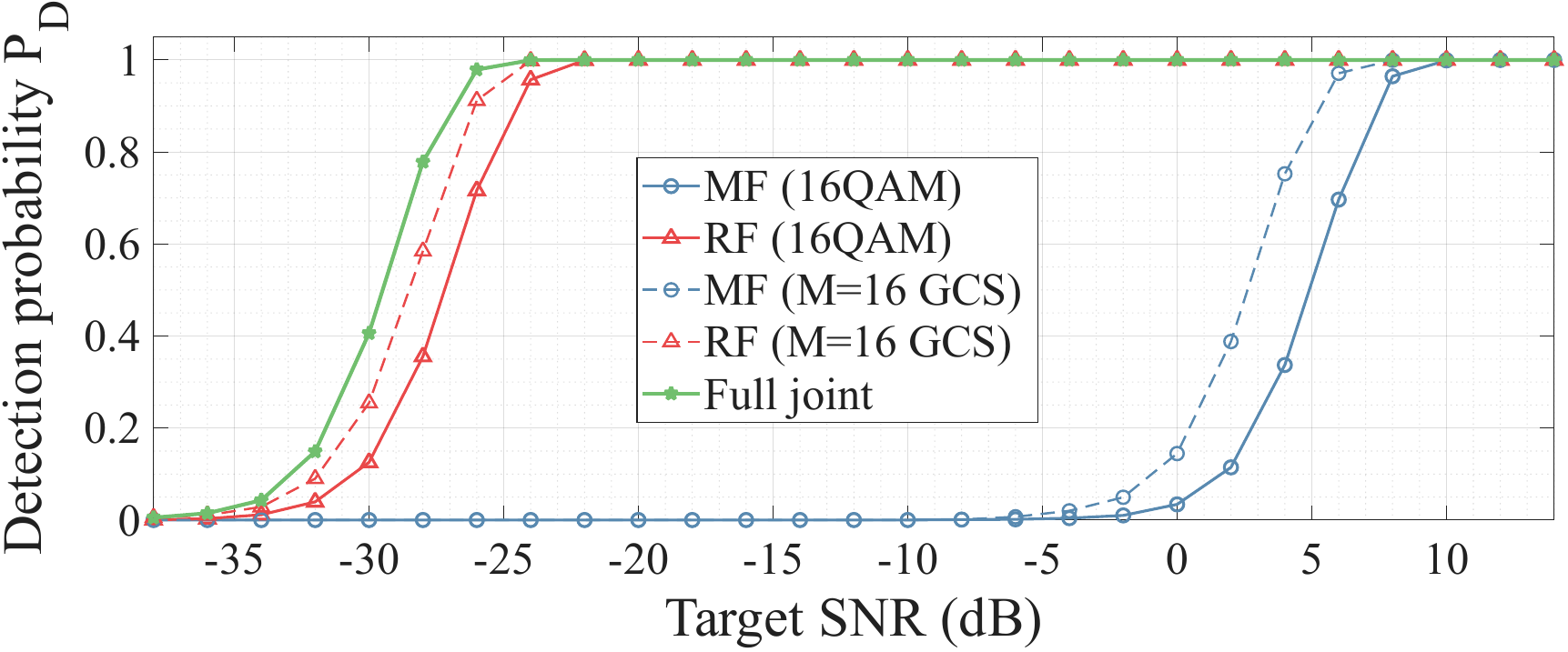}
    \caption{Weak-target detection probability under CFAR.}
    \label{fig:pd}
\end{figure}

Finally, Fig.~\ref{fig:tradeoff} traces the overall sensing-communication trade-off that the rate constraint \eqref{eq:ratecon} imposes, by sweeping the shaped constellations through the estimation benchmark under a one-transmit-sweep equal-MI convention, in which the MF, the RF, and the proposed joint design receive the same transmitted sweep, drawn from the equal-MI shaped family of the proposed design, at all seven shaping levels. At the constant-modulus endpoint the three chains coincide, at $0.0233$ and $0.0231$~m for $L=16$ and $64$, where the proposed design releases its boost, $\gamma\to\eta$ at $\kappa_2=0$. Toward the communication-optimal endpoint the MF more than doubles its RMSE and the RF incurs a growing noise penalty, while the proposed joint design remains lowest at every level, with margins over the RF of $13\%$ and $22\%$ for $L=16$ and $64$. The growth of the margin with the order arises because higher-order shaped payloads place points nearer the origin, so that $\nu_{-2}$ grows with the order, consistent with Lemma~\ref{lem:noise}. These results confirm that the proposed joint design extends the attainable sensing-communication trade-off region beyond that of any fixed transmit-receive combination.

\vspace{-1\baselineskip}
% =====================================================================
\section{Conclusion}\label{sec:concl}

This paper presented a joint transmit-receive design framework for full-frame OFDM-ISAC, in which the modulation constellation, the pilot placement and power split, and the receive filter are designed as a single optimization under a prescribed communication-rate requirement. The analysis showed that the payload-induced interference floor of the proposed ROI-MMF is governed by a single statistic, the effective centered second moment. Guided by this result, the receive filter is obtained in closed form, the balanced staggered pilot placement removes the deterministic grating lobes without sacrificing the channel-estimation accuracy, and the constellation-power operating point is selected on the rate boundary. Simulations confirmed that the closed-form expressions match the measurements, that weak targets masked under conventional processing are restored to clear detectability at the declared QoS level, and that the ranging and velocity estimation accuracy closely tracks the CRB. The entire OFDM frame, pilots and data payloads alike, can therefore serve as a high-quality sensing reference at a small and declared communication-rate cost.
\begin{figure}[t!]
    \centering
    \begin{minipage}{0.49\linewidth}\centering
        \includegraphics[width=\linewidth]{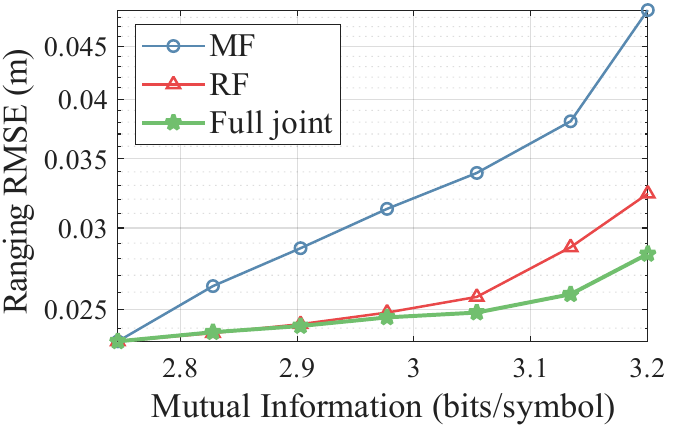}\\[2pt]
        {\footnotesize (a)}
    \end{minipage}\hfill
    \begin{minipage}{0.49\linewidth}\centering
        \includegraphics[width=\linewidth]{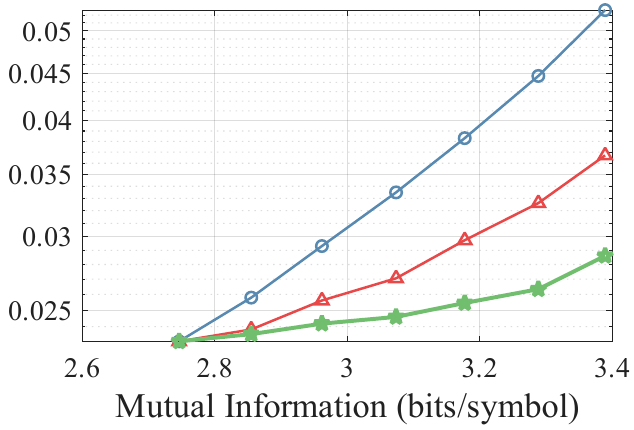}\\[2pt]
        {\footnotesize (b)}
    \end{minipage}
    \caption{Sensing-communication trade-off between the ranging RMSE
    and the payload mutual information for (a) $L=16$ and (b) $L=64$.}
    \label{fig:tradeoff}
\end{figure}
\vspace{-1\baselineskip}
% =====================================================================

\section*{Appendix A: Proof of Lemma~\ref{lem:keff}}
\label{app:lemma2}

Throughout, $\bm\phi_\Delta$ denotes $\bm\phi(q,k)$ for a lag $\Delta=(q,k)$.

Substituting \eqref{eq:split} into \eqref{eq:probe} gives
\begin{equation}
\begin{aligned}
    \chi(\Delta)
    &=\frac{\alpha_c}{D}\sum_{n=0}^{D-1}\bar p_n e^{-j\theta_\Delta(n)}
    +\chi_\psi+\chi_w, \\[3pt]
    \chi_\psi&=\frac{\alpha_c}{D}\sum_{n\in\mathcal D}
    \psi_n\,e^{-j\theta_\Delta(n)},
\end{aligned}
    \label{eq:app_split}
\end{equation}
where the deterministic sum vanishes off the comb harmonics by Lemma~\ref{lem:mean}. Since the $\psi_n$ are i.i.d. and zero-mean by \eqref{eq:fstats},
\begin{equation}
    \mathbb E\big[|\chi_\psi|^2\big]
    =\frac{|\alpha_c|^2}{D^2}\sum_{n\in\mathcal D}\mathrm{Var}(\psi_n)
    =\frac{|\alpha_c|^2}{D^2}\,D_d\,p_d^2\kappa_2
    =\frac{|\alpha_c|^2\keff}{D},
    \label{eq:app_floor}
\end{equation}
with $D_d/D=1-\eta$, and the filtered-noise power is $\mathbb E[|\chi_w|^2]=\sigma^2\|\mathbf h\|^2/D^2=\sigma^2\delta_N/D$. Adding the two gives \eqref{eq:keff_mf}. For the RF, $a_n=1$ in \eqref{eq:gain_fact}, so the noise term alone gives \eqref{eq:keff_rf}. The pilot cells satisfy $\psi_n=0$ by their constant modulus, which proves the pilot-exclusion claim.

Write \eqref{eq:woodbury} as $\mathbf h_{\rm ROI}=\beta(\mathbf x-\mathbf Z\mathbf b)/\lambda$ with $\mathbf b=(\lambda\mathbf I+\mathbf Z^H\mathbf Z)^{-1}\mathbf Z^H\mathbf x$. The covered-lag responses are the entries of
\begin{equation}
    \mathbf Z^H\mathbf h_{\rm ROI}
    =\frac{\beta}{\lambda}\big(\mathbf Z^H\mathbf x
    -\mathbf Z^H\mathbf Z\,\mathbf b\big)
    =\beta\,\big(\lambda\mathbf I+\mathbf Z^H\mathbf Z\big)^{-1}
    \mathbf Z^H\mathbf x.
    \label{eq:app_covered}
\end{equation}
With $\mathbf Z^H\mathbf Z\approx D\,\mathbf I$ and $\beta\approx\lambda$ (Appendix~B), and $\mathbb E\big|[\mathbf Z^H\mathbf x]_\Delta\big|^2=D\keff$ by \eqref{eq:app_floor},
\begin{equation}
    \mathbb E\big|\mathbf h_{\rm ROI}^H
    (\mathbf x\odot\bm\phi_\Delta)\big|^2
    \approx\Big(\frac{\lambda}{\lambda+D}\Big)^{2}D\keff,
    \qquad \Delta\in\mathcal S,
    \label{eq:app_covered_level}
\end{equation}
which, combined with the filtered-noise power above, yields \eqref{eq:keff_roi}. At an uncovered lag $\Delta\notin\mathcal S$,
\begin{equation}
    \mathbf h_{\rm ROI}^H(\mathbf x\odot\bm\phi_\Delta)
    =\frac{\beta}{\lambda}\Big[\mathbf x^H(\mathbf x\odot\bm\phi_\Delta)
    -\sum_{i\in\mathcal S}b_i^{*}\,
    \mathbf z_i^H(\mathbf x\odot\bm\phi_\Delta)\Big],
    \label{eq:app_uncovered}
\end{equation}
where the first term has power $D\keff$ by \eqref{eq:app_floor}. In the second, $\mathbb E|b_i|^2\approx D\keff/(\lambda+D)^2$ by \eqref{eq:app_covered} and the $|\mathcal S|$ summands are uncorrelated, so
\begin{equation}
    \mathbb E\Big|\sum_{i\in\mathcal S}b_i^{*}\,
    \mathbf z_i^H(\mathbf x\odot\bm\phi_\Delta)\Big|^2
    \!\!\approx\!|\mathcal S|\,\frac{D\keff}{(\lambda+D)^2}\,D\,\keff
    \!\!\approx\!|\mathcal S|\,(\keff)^2,
    \label{eq:app_crossterm}
\end{equation}

whose ratio to the direct term $D\keff$ is $|\mathcal S|\keff/D$, so the uncovered floor is unchanged to first order. 
\hfill$\blacksquare$
\vspace{-0.7\baselineskip}
\section*{Appendix B: Proof of Lemma~\ref{lem:noise}}
\label{app:lemma3}

\emph{1. MF:} The unit-gain normalization gives
\begin{equation}
    \delta_N^{\rm MF}=\frac{D}{\|\mathbf x\|^2},
    \qquad
    \mathbb E\|\mathbf x\|^2=D,\quad
    \mathrm{Var}\big(\|\mathbf x\|^2\big)=D\keff,
    \label{eq:app_mf_noise}
\end{equation}
so $\delta_N^{\rm MF}=1$ up to the relative fluctuation $\sqrt{\keff/D}$.

\emph{2. RF:} With $h_n=1/x_n^{*}$, the constraint $\mathbf h^H\mathbf x=D$ holds exactly and
\begin{equation}
    \mathbb E\|\mathbf h\|^2
    =\sum_{n\in\mathcal P}\frac{1}{p_p}
    +\sum_{n\in\mathcal D}\frac{1}{p_d}
    \mathbb E\Big[\frac{1}{|s_n|^2}\Big]
    =\frac{D_p}{p_p}+\frac{D_d}{p_d}\nu_{-2},
    \label{eq:app_rf_noise}
\end{equation}
which is \eqref{eq:deltaN_rf}.

\emph{3. ROI-MMF:} The constrained minimizer of \eqref{eq:roimmf} is $\mathbf h=\mu(\lambda\mathbf I+\mathbf Z\mathbf Z^H)^{-1}\mathbf x$ with $\mu$ fixed by $\mathbf h^H\mathbf x=D$, so that, with $\mathbf A=(\lambda\mathbf I+\mathbf Z\mathbf Z^H)^{-1}$,
\begin{equation}
    \delta_N=\frac{\|\mathbf h\|^2}{D}
    =\frac{D\,\mathbf x^H\mathbf A^2\mathbf x}
    {(\mathbf x^H\mathbf A\mathbf x)^2}.
    \label{eq:app_roi_quot}
\end{equation}
Split $\mathbf x$ into its component in the constraint span and the orthogonal remainder. The Gram concentrates on a scaled identity,
\begin{equation}
    [\mathbf Z^H\mathbf Z]_{ii}=\|\mathbf x\|^2\approx D,
    \quad
    \mathbb E\big|[\mathbf Z^H\mathbf Z]_{ij}\big|^2=D\keff,
    \,\text{\small$i\neq j$},
    \label{eq:app_gram}
\end{equation}

since each off-diagonal entry is a fluctuation sum of \eqref{eq:app_floor}. Likewise $\mathbb E|\mathbf z_i^H\mathbf x|^2/D=\keff$, so the span holds $\keff|\mathcal S|$ and the remainder $D-\keff|\mathcal S|$. For the boosted frame, the harmonic lags $\mathcal S_{\rm harm}$ appended to $\mathcal S$ add the mask-ambiguity energy of \eqref{eq:mask_amb}, bounded by
\begin{equation}
    \frac1D\sum_{\Delta\in\mathcal S_{\rm harm}}
    \Big|\sum_{n=0}^{D-1}\bar p_n e^{-j\theta_\Delta(n)}\Big|^2
    \le (p_p-p_d)^2\,\eta(1-\eta)\,D,
    \label{eq:app_parseval}
\end{equation}
since the total over all nonzero lags equals $D\sum_{n=0}^{D-1}(\bar p_n-1)^2=(p_p-p_d)^2\eta(1-\eta)D^2$ by Parseval's theorem, and this addition is negligible at the operating boosts. Since $\mathbf A$ passes the remainder with gain $1/\lambda$ and the span with $1/(\lambda+D)$, both quadratic forms in \eqref{eq:app_roi_quot} are dominated by the remainder,
\begin{equation}
    \mathbf x^H\mathbf A\mathbf x\approx\frac{D-\keff|\mathcal S|}{\lambda},
    \qquad
    \mathbf x^H\mathbf A^2\mathbf x\approx\frac{D-\keff|\mathcal S|}{\lambda^2},
    \label{eq:app_quadforms}
\end{equation}
which, substituted into \eqref{eq:app_roi_quot}, gives
\begin{equation}
    \delta_N^{\rm ROI}=\frac{D}{D-\keff|\mathcal S|}
    \label{eq:app_roi_final}
\end{equation}
up to corrections smaller by $|\mathcal S|/D$, which is \eqref{eq:deltaN_roi}. Moreover, \eqref{eq:app_quadforms} gives the constant $\beta\approx\lambda$ used in Appendix~A. 
\hfill$\blacksquare$
\vspace{-0.7\baselineskip}
\section*{Appendix C: Proof of Theorem~\ref{thm:decomp}}
\label{app:thm1}

Fix the transmit design and a unit-gain filter $\mathbf h(\mathbf x)=\tilde\beta\,(\mathbf x\odot\mathbf w)$, so that $h_n^*x_n=\tilde\beta^{*}w_n^{*}|x_n|^2$. The MF has $\mathbf w=\mathbf 1$ and the ROI-MMF has $w_n=1-\xi_n$ with $\xi_n$ the notch correction \eqref{eq:notch_coef}, whose mean square is $|\mathcal S|\keff/D$ to first order. For a shift $\Delta\in\mathcal S$, define
\begin{equation}
    S_\Delta=\mathbf h^H\big(\mathbf x\odot\bm\phi(\Delta)\big)
    =\sum_{n=0}^{D-1} h_n^{*}x_n\,e^{-j\theta_\Delta(n)}.
    \label{eq:app_stat}
\end{equation}
Splitting each term into its mean and centered parts, $\mathbb E|S_\Delta|^2=|\mathbb E S_\Delta|^2+\mathrm{Var}(S_\Delta)$, and summing over $\mathcal S$ yields
\begin{equation}
\mathbb E\sum_{\Delta\in\mathcal S}\big|S_\Delta\big|^2
=\sum_{\Delta\in\mathcal S}
\Big|\sum_{n=0}^{D-1} \bar u_n e^{-j\theta_\Delta(n)}\Big|^2
+\keff\,\Omega_{\mathbf h}(\mathcal S),
\label{eq:app_decomp}
\end{equation}
where the payload weight is defined for $\keff>0$ as
\begin{equation}
\Omega_{\mathbf h}(\mathcal S)
=\frac{1}{\keff}\sum_{\Delta\in\mathcal S}
\mathrm{Var}\big(S_\Delta\big),
\label{eq:omega_def}
\end{equation}
where the sum collects the payload variance of the covered lags. Adding $\lambda\,\mathbb E\|\mathbf h\|^2$ to \eqref{eq:app_decomp} gives \eqref{eq:explicit_obj}.

To first order in $|\mathcal S|/D$, inserting \eqref{eq:split} into $\bar u_n=\mathbb E[h_n^{*}x_n]$ gives
\begin{equation}
    \bar u_n
    =\bar p_n\,\mathbb E[\tilde\beta^{*}w_n^{*}]
    +\mathbb E[\tilde\beta^{*}w_n^{*}\psi_n],
    \label{eq:app_ubar}
\end{equation}
whose second term is smaller by $|\mathcal S|/D$, so $\bar u_n$ reduces to $\bar p_n$ and the deterministic term of \eqref{eq:explicit_obj} vanishes by Proposition~\ref{prop:balanced} outside the nonzero-Doppler harmonics.

With $\mathbf w=\mathbf 1$ and $\tilde\beta=D/\|\mathbf x\|^2\approx1$,
\begin{equation}
    \mathrm{Var}(S_\Delta)
    =\sum_{n\in\mathcal D}\mathrm{Var}(\psi_n)
    =D\keff
    \;\Longrightarrow\;
    \Omega_{\mathbf h}^{\rm MF}=D\,|\mathcal S|.
    \label{eq:app_omega_mf}
\end{equation}
For the ROI-MMF, inserting \eqref{eq:woodbury} with $\mathbf Z^H\mathbf Z\approx D\mathbf I$ and $\beta\approx\lambda$, as in \eqref{eq:app_covered}, leaves
\begin{equation}
\begin{aligned}
    \mathbf h^H(\mathbf x\odot\bm\phi_\Delta)
    &=\frac{\lambda}{\lambda+D}
    \sum_{n\in\mathcal D}\psi_n e^{-j\theta_\Delta(n)} \\
    \Longrightarrow\;\;
    \Omega_{\mathbf h}^{\rm ROI}
    &=\Big[\frac{\lambda}{\lambda+D}\Big]^{2}D\,|\mathcal S|,
\end{aligned}
    \label{eq:app_omega_roi}
\end{equation}
which establishes \eqref{eq:omega_vals}. The neglected terms are smaller by $|\mathcal S|/D$, and the leading payload term is a positive multiple of $\keff$, so the objective is strictly increasing in $\keff$ to the same accuracy. 
\hfill$\blacksquare$

\vspace{-0.7\baselineskip}
\section*{Appendix D: Proof of Proposition~\ref{prop:balanced}}
\label{app:prop1}
(i) By \eqref{eq:balanced_class} every subcarrier carries the same number $c_p=D_p/N=\eta M$ of pilots, so that for every $n$,
\begin{equation}
\begin{aligned}
    \sum_{m=0}^{M-1}\mathbb E|X_{n,m}|^2
    &=c_pp_p+(M-c_p)p_d \\
    &=M\big(\eta p_p+(1-\eta)p_d\big)=M
\end{aligned}
    \label{eq:flat_profile}
\end{equation}
by \eqref{eq:frame_power_constraint}, independent of $n$ and $\gamma$. The zero-Doppler deterministic response is the DFT of this flat profile, namely a single spike at lag zero. (ii) For taps $\ell,\ell'$, the balanced count makes the pilot subcarrier multiset equal to $\{0,\ldots,N-1\}$ repeated $c_p$ times, so that
\begin{equation}
\begin{aligned}
    [\bm\Omega^H\bm\Omega]_{\ell,\ell'}
    &=p_p\!\!\sum_{(n,m)\in\mathcal P}\!\!e^{j2\pi(\ell-\ell')n/N} \\
    &=p_pc_p\sum_{n=0}^{N-1}e^{j2\pi(\ell-\ell')n/N}
    =p_pD_p\,\delta_{\ell,\ell'},
\end{aligned}
    \label{eq:gram_balanced}
\end{equation}
which completes the proof.

\vspace{-0.7\baselineskip}
\section*{Appendix E: Proof of Proposition~\ref{prop:opt}}
\label{app:prop2}

Let $\mathbf h_{\rm ROI}$ denote the per-realization filter \eqref{eq:woodbury} evaluated at $\mathbf x$. Fix a pair $(\mathcal X,\gamma)$ satisfying \eqref{eq:ratecon}. For every placement $\mathcal P\in\mathcal B$ and every filter policy $\mathbf h(\cdot)$ obeying \eqref{eq:gaincon}, the objective of $(\mathcal{P}_1)$ satisfies
\begin{align}
    &\mathbb E\Big[\lambda\|\mathbf h(\mathbf x)\|^2
      +\!\sum_{\Delta\in\mathcal S}\!\big|\mathbf h(\mathbf x)^{H}
      (\mathbf x\odot\bm\phi_\Delta)\big|^{2}\Big]\nonumber\\
    &\overset{(a)}{\ge}\;
      \lambda\,\mathbb E\|\mathbf h_{\rm ROI}\|^2
      +\mathbb E\!\sum_{\Delta\in\mathcal S}\!\big|\mathbf h_{\rm ROI}^{H}
      (\mathbf x\odot\bm\phi_\Delta)\big|^{2}\nonumber\\
    &\overset{(b)}{=}\;
      \lambda\,\mathbb E\|\mathbf h_{\rm ROI}\|^2
      +\!\sum_{\Delta\in\mathcal S}\Big|\sum_{n=0}^{D-1}\bar u_n[\bm\phi_\Delta]_n\Big|^2
      +\keff\,\Omega_{\mathbf h}(\mathcal S)\nonumber\\
    &\overset{(c)}{\ge}\;
      \lambda\,\mathbb E\|\mathbf h_{\rm ROI}\|^2+\keff\,\Omega_{\mathbf h}(\mathcal S)
      \; =: \;\mathcal J(\mathcal X,\gamma).
    \label{eq:nested}
\end{align}
Step (a) holds since the expectation decouples across realizations, so the minimum over $\mathbf h(\cdot)$ is the per-realization ROI-MMF \eqref{eq:roimmf}--\eqref{eq:woodbury}. Step (b) is \eqref{eq:app_decomp}. Step (c) drops the deterministic term, which vanishes by Proposition~\ref{prop:balanced}(i) at every lag of $\mathcal S$ except the nonzero-Doppler harmonics, where it is bounded by \eqref{eq:app_parseval} and suppressed by the placement cost \eqref{eq:mask_opt}, which serves as a surrogate for this residual. Hence $\mathcal J$ is placement-independent to the stated accuracy.

By \eqref{eq:app_omega_roi} and Lemma~\ref{lem:noise}, both terms of $\mathcal J$ depend on the transmit variables only through $\keff$ and increase with it, so minimizing $\mathcal J$ subject to \eqref{eq:ratecon} is exactly \eqref{eq:joint_opt}. Writing \eqref{eq:sinreff} as
\begin{equation}
    \mathrm{SINR}_{\rm eff}
    =\big[\mathrm{tr}(\mathbf C_\epsilon)+\sigma_c^2/p_d\big]^{-1},
    \label{eq:app_sinr_inv}
\end{equation}
the denominator is convex in $\gamma$, so $\mathrm{SINR}_{\rm eff}$ is quasiconcave, $R$ is quasiconcave through the monotone $I_{\mathcal X}$, and the set feasible for \eqref{eq:ratecon} is an interval. The monotone $\keff$ therefore attains its minimum at the largest feasible boost, and the rate and $\keff$ orderings of the levels select the tightest feasible level. The resulting $(\mathcal X^\star,\gamma^\star)$, together with $\mathcal P^\star$ and $\mathbf h_{\rm ROI}$, is the output of the three steps. This completes the proof.
\hfill$\blacksquare$
\bibliographystyle{IEEEtran}
\bibliography{IEEEabrv,reference}

@STRING{IEEE_J_AES        = "{IEEE} Trans. Aerosp. Electron. Syst."}

@STRING{IEEE_J_STSP       = "{IEEE} J. Sel. Topics Signal Process."}

@STRING{IEEE_J_SP         = "{IEEE} Trans. Signal Process."}

@STRING{IEEE_J_JSAC       = "{IEEE} J. Sel. Areas Commun."}

@STRING{IEEE_J_COM        = "{IEEE} Trans. Commun."}

@STRING{IEEE_J_WCOM       = "{IEEE} Trans. Wireless Commun."}

@STRING{IEEE_J_IT         = "{IEEE} Trans. Inf. Theory"}

@STRING{IEEE_J_IOT        = "{IEEE} Internet Things J."}

@STRING{IEEE_J_PROC       = "Proc. {IEEE}"}

@article{yang2026constellation,
  title={{Constellation-Independent Range Estimation in Payload-Based OFDM-ISAC}},
  author={Yang, Dongil and Meng, Kaitao and Masouros, Christos and Han, Kawon},
  journal={IEEE Wireless Communications Letters},
  year={2026},
  publisher={IEEE}
}

@article{xie2025adaptive,
  title={{Adaptive matched filtering for sensing with communication signals in cluttered environments}},
  author={Xie, Lei and He, Hengtao and Xiong, Yifeng and Liu, Fan and Jin, Shi},
  journal={arXiv preprint arXiv:2512.08157},
  year={2025}
}

@article{han2026next,
  title={{Next-generation MIMO transceivers for integrated sensing and communications: Unique security vulnerabilities and solutions}},
  author={Han, Kawon and Masouros, Christos and Riihonen, Taneli and Amin, Moeness G},
  journal={Proceedings of the IEEE},
  year={2026},
  publisher={IEEE}
}

@article{zhang2024cross,
  title={{Cross-domain dual-functional OFDM waveform design for accurate sensing/positioning}},
  author={Zhang, Fan and Mao, Tianqi and Liu, Ruiqi and Han, Zhu and Chen, Sheng and Wang, Zhaocheng},
  journal={IEEE Journal on Selected Areas in Communications},
  volume={42},
  number={9},
  pages={2259--2274},
  year={2024},
  publisher={IEEE}
}

@article{bouziane2026optimized,
  title={{Optimized Non-Uniform Pilot Pattern for OFDM Sensing}},
  author={Bouziane, Amir and Arslan, H{\"u}seyin},
  journal={IEEE Wireless Communications Letters},
  volume={15},
  pages={2463--2467},
  year={2026},
  publisher={IEEE}
}

@article{boutillon2024constellations,
  title={{Constellations Cross Circular auto-Correlation C4-sequences}},
  author={Boutillon, Emmanuel},
  journal={IEEE Transactions on Communications},
  volume={72},
  number={12},
  pages={7664--7673},
  year={2024},
  publisher={IEEE}
}

@article{morelli2001comparison,
  title={{A comparison of pilot-aided channel estimation methods for OFDM systems}},
  author={Morelli, Michele and Mengali, Umberto and others},
  journal={IEEE Transactions on signal processing},
  volume={49},
  number={12},
  pages={3065--3073},
  year={2001}
}

@article{liu2025probabilistic,
  title={{Probabilistic shaping based ISAC systems with finite constellations: Analysis and optimization}},
  author={Liu, Yichen and Guo, Yinghong and Gu, Yixiao and Wang, Manlin and Liu, Junhua and Xia, Bin},
  journal={IEEE Transactions on Wireless Communications},
  year={2025},
  publisher={IEEE}
}

@article{yang2024constellation,
  title={{Constellation design for integrated sensing and communication with random waveforms}},
  author={Yang, Xiaobo and Zhang, Ruonan and Zhai, Daosen and Liu, Fan and Du, Rui and Han, Tony Xiao},
  journal={IEEE Transactions on Wireless Communications},
  volume={23},
  number={11},
  pages={17415--17428},
  year={2024},
  publisher={IEEE}
}

@article{meng2026constellation,
  title={{Constellation Selection and Power Control for OFDM-based ISAC: From Theory to Prototype}},
  author={Meng, Kaitao and Han, Kawon and Masouros, Christos and Liu, Fan},
  journal={IEEE Transactions on Signal Processing},
  year={2026},
  publisher={IEEE}
}

@article{xu2025exploiting,
  title={{Exploiting both pilots and data payloads for integrated sensing and communications}},
  author={Xu, Chen and Yu, Xianghao and Liu, Fan and Jin, Shi},
  journal={IEEE Transactions on Wireless Communications},
  year={2025},
  publisher={IEEE}
}

@article{geiger2026constellation,
  title={{Constellation shaping for OFDM-ISAC systems: From theoretical bounds to practical implementation}},
  author={Geiger, Benedikt and Liu, Fan and Lu, Shihang and Rode, Andrej and Gaviria, Daniel Gil and Muth, Charlotte and Schmalen, Laurent},
  journal={IEEE Transactions on Communications},
  year={2026},
  publisher={IEEE}
}

@article{lu2024random,
  title={{Random ISAC signals deserve dedicated precoding}},
  author={Lu, Shihang and Liu, Fan and Dong, Fuwang and Xiong, Yifeng and Xu, Jie and Liu, Ya-Feng and Jin, Shi},
  journal={IEEE Transactions on Signal Processing},
  volume={72},
  pages={3453--3469},
  year={2024},
  publisher={IEEE}
}

@article{liu2022integrated,
	title={{Integrated sensing and communications: Towards dual-functional wireless networks for {6G} and beyond}},
	author={Liu, Fan and Cui, Yuanhao and Masouros, Christos and Xu, Jie and Han, Tony Xiao and Eldar, Yonina C and Buzzi, Stefano},
	journal=IEEE_J_JSAC,
	month={Jun.},
	year={2022},
	volume={40},
	number={6},
	pages={1728-1767},
	publisher={IEEE}
}

@article{zhang2021overview,
  title={{An overview of signal processing techniques for joint communication and radar sensing}},
  author={Zhang, J Andrew and Liu, Fan and Masouros, Christos and Heath, Robert W and Feng, Zhiyong and Zheng, Le and Petropulu, Athina},
  journal=IEEE_J_STSP,
  volume={15},
  number={6},
  pages={1295--1315},
  year={2021},
  publisher={IEEE}
}

@article{liu2020joint0,
  title={{Joint radar and communication design: Applications, state-of-the-art, and the road ahead}},
  author={Liu, Fan and Masouros, Christos and Petropulu, Athina P and Griffiths, Hugh and Hanzo, Lajos},
  journal=IEEE_J_COM,
  volume={68},
  number={6},
  pages={3834--3862},
  year={2020},
  publisher={IEEE}
}

@article{zhu2023pilot,
  title={{Pilot optimization for OFDM-based ISAC signal in emergency IoT networks}},
  author={Zhu, Wendi and Han, Yanhong and Wang, Li and Xu, Lianming and Zhang, Yuming and Fei, Aiguo},
  journal=IEEE_J_IOT,
  volume={11},
  number={18},
  pages={29600--29614},
  year={2023},
  publisher={IEEE}
}

@article{hua2024integrated,
  title={{Integrated sensing and communication: Joint pilot and transmission design}},
  author={Hua, Meng and Wu, Qingqing and Chen, Wen and Jamalipour, Abbas and Wu, Celimuge and Dobre, Octavia A},
  journal=IEEE_J_WCOM,
  year={2024},
  publisher={IEEE}
}

@article{liu2025uncovering,
  title={{Uncovering the iceberg in the sea: Fundamentals of pulse shaping and modulation design for random ISAC signals}},
  author={Liu, Fan and Xiong, Yifeng and Lu, Shihang and Li, Shuangyang and Yuan, Weijie and Masouros, Christos and Jin, Shi and Caire, Giuseppe},
  journal=IEEE_J_SP,
  year={2025},
  publisher={IEEE}
}

@article{liu2025cp,
  title={{CP-OFDM achieves the lowest average ranging sidelobe under QAM/PSK constellations}},
  author={Liu, Fan and Zhang, Ying and Xiong, Yifeng and Li, Shuangyang and Yuan, Weijie and Gao, Feifei and Jin, Shi and Caire, Giuseppe},
  journal=IEEE_J_IT,
  year={2025},
  publisher={IEEE}
}

@article{sturm2011waveform,
  title={{Waveform design and signal processing aspects for fusion of wireless communications and radar sensing}},
  author={Sturm, Christian and Wiesbeck, Werner},
  journal=IEEE_J_PROC ,
  volume={99},
  number={7},
  pages={1236--1259},
  year={2011},
  publisher={IEEE}
}

@article{keskin2025fundamental,
  title={{Fundamental trade-offs in monostatic ISAC: A holistic investigation towards 6G}},
  author={Keskin, Musa Furkan and Mojahedian, Mohammad Mahdi and Lacruz, Jesus O and Marcus, Carina and Eriksson, Olof and Giorgetti, Andrea and Widmer, Joerg and Wymeersch, Henk},
  journal=IEEE_J_WCOM,
  year={2025},
  publisher={IEEE}
}

@article{zhang2025optimal,
  title={{Optimal Power Allocation for OFDM-Based Ranging Using Random Communication Signals}},
  author={Zhang, Ying and Liu, Fan and Liu, Tao and Jin, Shi},
  journal=IEEE_J_WCOM,
  volume={25},
  pages={6460--6473},
  year={2025},
  publisher={IEEE}
}

@article{du2024reshaping,
  title={{Reshaping the ISAC tradeoff under OFDM signaling: A probabilistic constellation shaping approach}},
  author={Du, Zhen and Liu, Fan and Xiong, Yifeng and Han, Tony Xiao and Eldar, Yonina C and Jin, Shi},
  journal=IEEE_J_SP,
  year={2024},
  publisher={IEEE}
}

@article{mcaulay1971optimal,
  title={{Optimal mismatched filter design for radar ranging, detection, and resolution}},
  author={McAulay, R and Johnson, J},
  journal=IEEE_J_IT,
  volume={17},
  number={6},
  pages={696--701},
  year={1971},
  publisher={IEEE}
}

@article{wojaczek2018reciprocal,
  title={{Reciprocal-filter-based STAP for passive radar on moving platforms}},
  author={Wojaczek, Philipp and Colone, Fabiola and Cristallini, Diego and Lombardo, Pierfrancesco},
  journal=IEEE_J_AES,
  volume={55},
  number={2},
  pages={967--988},
  year={2018},
  publisher={IEEE}
}

@article{rodriguez2023supervised,
  title={{Supervised reciprocal filter for OFDM radar signal processing}},
  author={Rodriguez, Javier Trujillo and Colone, Fabiola and Lombardo, Pierfrancesco},
  journal=IEEE_J_AES,
  volume={59},
  number={4},
  pages={3871--3889},
  year={2023},
  publisher={IEEE}
}

@article{han2025sensing,
  author={Han, Kawon and Meng, Kaitao and Masouros, Christos},
  title={{Sensing-Secure ISAC: Ambiguity Function Engineering for Impairing Unauthorized Sensing}}, 
  journal=IEEE_J_WCOM,
  year={2025},
  publisher={IEEE}
}

@article{simko2013adaptive,
  title={{Adaptive pilot-symbol patterns for MIMO OFDM systems}},
  author={Simko, Michal and Diniz, Paulo SR and Wang, Qi and Rupp, Markus},
  journal={IEEE Transactions on wireless communications},
  volume={12},
  number={9},
  pages={4705--4715},
  year={2013},
  publisher={IEEE}
}

@article{vsimko2012optimal,
  title={{Optimal pilot symbol power allocation under time-variant channels}},
  author={{\v{S}}imko, Michal and Wang, Qi and Rupp, Markus},
  journal={EURASIP Journal on Wireless Communications and Networking},
  volume={2012},
  number={1},
  pages={225},
  year={2012},
  publisher={Springer}
}

@inproceedings{han2026constellation,
  title={{Constellation design in OFDM-ISAC over data payloads: From MSE analysis to experimentation}},
  author={Han, Kawon and Meng, Kaitao and Chatzicharistou, Alexandra and Masouros, Christos},
  booktitle={ICC 2026-IEEE International Conference on Communications},
  pages={1--6},
  year={2026},
  organization={IEEE}
}

@article{keskin2021mimo,
  title={{MIMO-OFDM joint radar-communications: Is ICI friend or foe?}},
  author={Keskin, Musa Furkan and Wymeersch, Henk and Koivunen, Visa},
  journal=IEEE_J_STSP,
  volume={15},
  number={6},
  pages={1393--1408},
  year={2021},
  publisher={IEEE}
}

@article{luo2025isac,
  title={{ISAC--A Survey on Its Layered Architecture, Technologies, Standardizations, Prototypes and Testbeds}},
  author={Luo, Xuewen and Lin, Qingfeng and Zhang, Ruoyu and Chen, Hsiao-Hwa and Wang, Xingwei and Huang, Min},
  journal={IEEE Communications Surveys \& Tutorials},
  year={2025},
  publisher={IEEE}
}

@article{caire1998bit,
  title={Bit-interleaved coded modulation},
  author={Caire, Giuseppe and Taricco, Giorgio and Biglieri, Ezio},
  journal={IEEE transactions on information theory},
  volume={44},
  number={3},
  pages={927--946},
  year={1998},
  publisher={IEEE}
}

@article{xu2025does,
  title={{How does CP length affect the sensing range for OFDM-ISAC?}},
  author={Xu, Xiaoli and Zhou, Zhiwen and Zeng, Yong},
  journal=IEEE_J_SP,
  volume={73},
  pages={5106--5120},
  year={2025},
  publisher={IEEE}
}

@article{sarkar1995using,
  title={{Using the matrix pencil method to estimate the parameters of a sum of complex exponentials}},
  author={Sarkar, Tapan K and Pereira, Odilon},
  journal={IEEE Antennas and propagation Magazine},
  volume={37},
  number={1},
  pages={48--55},
  year={1995},
  publisher={IEEE}
}

@article{wu2026ambiguity,
  title={{Ambiguity Function Analysis of Pilot-Embedded Random OFDM Signals}},
  author={Wu, Jialin and Liu, Fan and Zhang, Ying and Xiong, Yifeng and Yang, Jie and Jin, Shi},
  journal={arXiv preprint arXiv:2607.17663},
  year={2026}
}

@techreport{3gpp38211,
  author      = {{3GPP}},
  title       = {{NR; Physical Channels and Modulation}},
  institution = {3rd Generation Partnership Project (3GPP)},
  number      = {TS 38.211},
  version     = {15.4.0},
  year        = {2019},
  month       = apr
}
\end{document}